\documentclass[journal]{IEEEtran}
\usepackage{caption}
\ifCLASSINFOpdf
\else
   \usepackage[dvips]{graphicx}
\fi
\usepackage{booktabs}
\usepackage{url}
\usepackage{float}
\usepackage{amsmath}
\usepackage{amssymb}
\usepackage{graphicx}
\usepackage{color}
\usepackage[table]{xcolor}
\definecolor{khaki}{RGB}{240,230,140}
\usepackage{url}
\usepackage[nospace,compress]{cite}
\usepackage{amsbsy}
\usepackage{epsfig}
\usepackage{subcaption}
\usepackage{multirow}

\makeatletter

\def\bx{{\mathbf x}}

\def\b0{{\mathbf 0}}

\newcommand{\beq}{\begin{equation}}
\newcommand{\eeq}{\end{equation}}

\def\ba{\mbox{\boldmath $a$}}

\def\bn{\mbox{\boldmath $n$}}

\def\by{\mbox{\boldmath $y$}}

\def\bx{\mbox{\boldmath $x$}}

\def\bs{\mbox{\boldmath $s$}}

\def\by{\mbox{\boldmath $y$}}

\def\mx{\mbox{$\mathbf{x}$}}
\def\mb{\mbox{$\mathbf{b}$}}

\def\mB{\mbox{$\mathbf{B}$}}

\def\mX{\mbox{$\mathbf{X}$}}

\def\mI{\mbox{$\mathbf{I}$}}
\def\mL{\mbox{$\mathbf{L}$}}

\def\mU{\mbox{$\mathbf{U}$}}

\def\ms{\mbox{$\mathbf{s}$}}
\newcommand{\ds}{\displaystyle}

\newtheorem{theorem}{\textbf{Theorem}}
\newtheorem{proposition}{Proposition}
\newtheorem{definition}{Definition}

\newenvironment{proof}[1][Proof]{\noindent \textbf{#1.} }{\qedsymbol}
\newcommand{\qedsymbol}{\hspace{\fill}\rule{1.5ex}{1.5ex}}
\floatstyle{ruled}
\newfloat{algorithm}{tbp}{loa}
\providecommand{\algorithmname}{Algorithm}
\floatname{algorithm}{\protect\algorithmname}

\usepackage{amssymb}
\begin{document}

\title{Topological Signal Processing Over Cell MultiComplexes Via Cross-Laplacian Operators}

\author{Stefania Sardellitti, \IEEEmembership{Senior Member, IEEE},  Breno C. Bispo, Fernando A. N. Santos, Juliano B. Lima, \IEEEmembership{Senior Member, IEEE}
\thanks{
This work was supported  by the  FIN-RIC Project 
TSP-ARK, financed by Universitas Mercatorum under grant n. 20-FIN/RIC. Additional support was provided by CAPES (88881.311848/2018-01, 88887.899136/2023-00), CNPq (442238/2023-1, 312935/2023-4, 405903/2023-5), and FACEPE (APQ-1226-3.04/22). This paper will be presented in part at the 33rd European Signal Processing Conference (EUSIPCO) 2025, Palermo, Italy. Sardellitti is with the Dept. of Engineering and Sciences, Universitas Mercatorum, Piazza Mattei 10, 00186, Rome, Italy. Bispo  and Lima are with Dept. of Electronics and Systems, Federal University of Pernambuco, Recife, Brazil. Santos is with the  Dutch Institute for Emergent Phenomena, KdVI, University of Amsterdam, Amsterdam, The Netherlands. E-mails: stefania.sardellitti@unimercatorum.it,\{breno.bispo, juliano.lima\}@ufpe.br, f.a.nobregasantos@uva.nl.}}

\maketitle

\begin{abstract}
One of the key challenges in many research fields is uncovering how different interconnected systems interact within complex networks,  typically represented as multi-layer networks. Capturing the intra- and cross-layer interactions among different domains for analysis and processing calls for topological algebraic descriptors capable of localizing the homologies of different domains, at different scales,  according to the learning task. Our first contribution in this paper is to introduce the Cell MultiComplexes (CMCs), which are novel topological spaces that enable the representation of higher-order interactions among interconnected cell complexes. We introduce cross-Laplacian operators as powerful algebraic descriptors of CMC spaces able to capture different topological invariants, whether global or local, at different resolutions. Using the eigenvectors of these operators as bases for the signal representation, we develop topological signal processing tools for signals defined over CMCs. Then, we focus on the signal spectral representation and on the filtering of noisy flows observed over the cross-edges between different layers of CMCs. We show that a local signal representation based on cross-Laplacians  yields a better sparsity/accuracy trade-off compared to monocomplex representations, which provide overcomplete representation of local signals. Finally, we illustrate a topology learning strategy designed to infer second-order cross-cells between layers, with applications to brain networks for encoding inter-module connectivity patterns.  




 \end{abstract}


\IEEEpeerreviewmaketitle

\section{Introduction}

 The study of complex networks has recently emerged as a vibrant area of research, offering powerful tools for analyzing  complex relationships within heterogeneous systems \cite{boccaletti2006complex},\cite{strogatz2001exploring}. Complex networks  are composed of multiple interconnected subsystems  that interact  through  relationships  having    different meanings  and often operating at different scales. %
Typically, these networks are composed of heterogeneous  domains organized in different layers of connectivity. 

In the last decades multilayer networks  \cite{kivela2014multilayer},\cite{boccaletti2006complex}, \cite{de2013mathematical}, \cite{bianconi2018multilayer} have gained a lot of research interest  due to their ability to model complex systems. Unlike common single-layer networks, multilayer networks account for relationships within and  among multiple layers of connectivity, where each layer may represent a sub-network, a different domain or a distinct snapshot of the same domain. 
Many human-made networks, such as power, telecommunication, social and transportation networks, as well as a plethora of natural phenomena, exhibit a sophisticated, highly interdependent structure well described by multilayer networks. For example, in telecommunication and transportation networks \cite{CRAINIC20221}, \cite{racz2024multilayer} multilayer networks are efficient tools to select physical and logical paths and optimize flows.
Social networks \cite{dickison2016multilayer} are one of the prominent examples of multilayer networks, where social entities are linked due to a social tie and each layer represents a different type of relationship.
In neuroscience, the hierarchical structure of the brain connectomes can be properly modeled by multilayer networks \cite{10.1093/gigascience/gix004}, where different interconnected layers correspond to distinct modes of brain connectivity, potentially able to capture the interplay among modules better  than a monolayer brain perspective \cite{breedt2023multimodal}. In biological molecular networks \cite{liu2020robustness},\cite{zhao2025constructing}, multilayer networks are suitable tools for modeling multiple biochemical interactions, as 
protein-gene-metabolite interactions.

Multilayer networks are modeled through graphs where the nodes within  and among distinct layers are connected through intra- and inter-layer  edges, respectively. However, many complex networks exhibit not just simple dyadic relations between entities, but   higher-order interactions involving  groups of entities.
Hence, since multilayer graphs are able to capture only pairwise relationships between couple of nodes, they fail to represent interactions involving groups of nodes. 
Recently, in \cite{krishnagopal2023topology} the authors introduced multiplex simplicial networks where  each
layer of the  network includes higher-order interactions modeled through simplicial complexes. Simplicial complexes \cite{munkres2018elements} are topological spaces able to capture higher-order interactions among the elements of a set under the inclusion property, i.e., if a set belongs to the complex, then all its subsets  also  belong to the complex.
These topological domains are  algebraically  represented through the so called higher-order Hodge Laplacian matrices \cite{Lim} which are  algebraic descriptors able to capture global invariants of the space, i.e. properties that keep unchanged under homeomorphic transformations of the space.

In \cite{barb_2020} the framework Topological Signal Processing (TSP) was introduced for the processing of signals defined over simplicial (mono-)complexes  by extending the classical signal processing tools such as spectral representations, filtering, sampling and recovering to data observed over topological spaces represented using higher-order Hodge Laplacians.  This framework was extended to cell complex structures in \cite{sardellitti2024topological}. In \cite{Roddenberry_23} the authors proposed a framework for signal processing on product spaces of
simplicial and cellular complexes leveraging the structure of  the Hodge Laplacian of the product space to jointly
filter along time and space.
Nevertheless, in multilayer networks, the Hodge Laplacians provide a global representation of the  space that fails to  localize homology across layers. Recently, the authors in \cite{moutuou2023} introduced an interesting representation of simplicial multi-complex networks based on the so called cross-Laplacian operators.
These algebraic descriptors provide a lens of different resolution for observing a multi-complex network composed of distinct layers where both the intra- and inter-layer interactions are modeled using simplicial complexes.
Interestingly, the cross-Laplacians are powerful algebraic tools to capture local or global topological invariants encoded by the so-called cross-Betti vectors.\\ 
Building a representation of the signals based on the cross-Laplacians, our goal in this paper is to extend topological signal processing to  novel layered topological higher-order domains that we named Cell MultiComplexes (CMCs), able of capturing relationships of any sparsity order among data.  
Our main contributions can be summarized as follows: 
\begin{itemize}
    \item[1)] We first introduce the novel CMC topological spaces that are a collection of cell complexes, each associated with a  layer,   interconnected by higher-order topological spaces named cross-complexes. Hence, we extend the algebraic representation of simplicial multi-complexes based on cross-Laplacians developed in \cite{moutuou2023} to  cell multi-complexes structures.
    The cross-Laplacians  enable the extraction of the local or global topological invariants according to the scale we aim to explore: a global perspective handling the entire complex as a flattened monolayer structure or a local lens which disentangles the homologies to study as the topology of a layer is related to the others. In this first study, we focus on $(0,0)$ cross-Laplacians.
    \item[2)]  We show how  signals observed on CMC spaces  admit a Hodge-based decomposition  in three orthogonal components. These components provide a local physical-based interpretation of the solenoidal, irrotational and harmonic edge signals enabling the introduction of cross-divergence and cross-curl operators. Using the eigenvectors of the cross-Laplacians as signal bases, we show how they enable a low-dimensional spectral representation of local signals by avoiding the overcomplete representation derived by a monocomplex approach. Then, we    
    extend the TSP framework developed in \cite{sardellitti2024topological} for a single cell-complex to CMCs spaces. We propose methods to find the optimal signal sparsity/accuracy trade-off and for  filtering  the signal components  from noisy observations. Furthermore, we  infer the structure   of $2$-order cross-complexes by considering an interesting real-data application to brain networks in order to learn inter-modules connectivity. 
\end{itemize}

Some preliminary results of our work were presented in  \cite{SardellittiCMC2024}. Here, we extend the work  in \cite{SardellittiCMC2024}  providing theoretical results for the representation of CMCs through cross-Laplacians  and for the cross-invariants, showing as sparse signal spectral representations can be derived, assessing the performance gain in terms of low-dimensional dictionaries when the processing of local signals  is required, proposing topology inference methods applied to brain networks.

The paper is organized as follows. In Section \ref{sec: Cell_multicomplexes} we 
introduce cross-complex topological spaces, while Section \ref{sec: Alg_rep} illustrates the algebraic framework for representing CMCs. 
In Section \ref{sec: Multilayer_graphs} we explore  multi-layer graphs from the perspective of cross-Laplacian-based representations and in Section \ref{sec: CMC_in_edge_space} we focus on  $2$-order cell multicomplexes by introducing  the cross-Betti numbers.
Section \ref{sec: signal_proc} provides the spectral representation of signals over CMCs,  finding the optimal trade-off between the sparsity of the signal representation and the data fitting error. In Section \ref{sec: signal_estim} we illustrate a method to estimate cross-edge signals, while in Section \ref{sec: signal_learning}  a learning strategy is presented to infer the structure of CMCs by validating its effectiveness in exploring the inter-modules connectivity in brain networks. Finally, in Section \ref{sec: conclusion} we draw some conclusions.
\vspace{-0.3cm}





\section{Cell MultiComplex Spaces}
\label{sec: Cell_multicomplexes}

In this section, we introduce the basic notions  defining Cell MultiComplex (CMC) spaces. 
We advance the topological tools developed in \cite{moutuou2023} for representing  simplicial complex-based multilayer networks, extending them to encompass more general and expressive topological structures, such as cell complexes. While simplicial complexes allow for the representation of relations of arbitrary order, they are constrained by the inclusion property, meaning that if a set belongs to the space, all of its subsets  must also be elements of the space. Relaxing this inclusion property, cell complexes are more general topological spaces which enable to capture relationships of any degree of sparsity  among elements of the space. 
Hence, we first recall the notion of cell complexes \cite{sardellitti2024topological},\cite{klette2000cell}.    Then, we introduce  cell multi-complexes that represent novel topological spaces defined as collections of cross-connected cell complexes. CMCs  provide a flexible framework  for capturing both intra- and inter-relations across interconnected domains organized into distinct layers of connectivity.
\vspace{-0.3cm}

\subsection{Cell complexes}
An  abstract cell complex (ACC) is a finite partially ordered set (poset), equipped with a dimension function,
 whose fundamental elements are called cells. Formally, a cell complex is defined as follows \cite{klette2000cell},\cite{sardellitti2024topological}.
 \begin{definition}\textit{
An abstract cell complex $\mathcal{C}=\{\mathcal{S},  \prec_b, \text{dim}\}$ is a set $\mathcal{S}$ of abstract elements $c$, named cells,
provided with a binary relation $\prec_b$, called the bounding (or incidence) relation, and with a dimension function, denoted by $\mbox{dim}(c)$, that assigns to each $c \in \mathcal{S}$ a non-negative integer $[c]$. Given the cells $c,w,z \in \mathcal{S}$,  the following two axioms are satisfied:
\begin{enumerate}
    \item[1)]  if $c \prec_b w$ and $w \prec_b z$, then $ c \prec_b z$  follows (transitivity);
    \item[2)] if $c \prec_b w$, then $\text{dim}(c)<\text{dim}(w)$ (monotonicity).
\end{enumerate}}
\end{definition}
A cell $c$ is called a $k$-cell if  the dimension (or order) of $c$ is $k$, i.e.
 $\text{dim}(c)=k$. 
  A cell of order $k$ is denoted by $c_k$. Therefore, 
 $0$-cells $c_0$ are named vertices and $1$-cells $c_1$ edges. With a slight abuse of notation, we refer to  $2$-order cells bounded by an arbitrary number of $1$-order cells as polygons, even though these are not necessarily embedded in a Euclidean space.\\ We say that the cell $c_k$ lower bounds the cell $c_{k+1}$ 
 if $ c_k \prec_b c_{k+1}$
and $c_k$ is  called a face of $c_{k+1}$. 
 An ACC is of dimension $K$, if the maximum dimension of  its cells is equal to $K$. Note, that an ACC of dimension $1$ is a graph.
 Given a $k$-order cell $c_k$,
 we define its boundary as the set of all cells of order less than $k$ that bound $c_k$. 
 An ACC equipped with neighboring relations among the cells of the complex is a
 topological space \cite{barmak2011algebraic}, i.e. a set of elements along with an ensemble of neighboring relations among them. A cell complex becomes  a simplicial complex (SC) if its cells are constrained to satisfy the inclusion property. Specifically, in simplicial complex spaces, a $k$-order cell is called a $k$-simplex.\\ 
In   Fig. \ref{fig:CMC_1}(a) we illustrate an example of an ACC of order $3$ composed of  triangles, 
squares, a  pentagon and a tetrahedron.

\vspace{-0.3cm}

\subsection{Cell MultiComplexes}
\label{Sec: CMC}
Let us now introduce the novel notion of cell multicomplex space.
\begin{definition}
\textit{A Cell MultiComplex (CMC) $\mathcal{X}$ is a topological space composed of a finite collection  of interdependent abstract cell complexes, each  associated with a topological layer. The interdependence among these complexes involve higher-order 
inter-layer interactions modeled by cross-complexes.}   \end{definition}
An illustrative example of a CMC is shown in  Fig. \ref{fig:CMC_1}(b). Although  it shares  the same  topological structure as the ACC in Fig.  \ref{fig:CMC_1}(a),  it is representative of data residing on different domains. Each layer can be associated, for instance, with a different network, a distinct domain, or a  snapshot of the same domain  at a different time, allowing the cell multicomplex to represent relationships among signals within or across domains.

The inter-layer higher-order interactions  are captured by cells of different orders  named \textit{cross-cells}.  Cross-cells of order $1$, $2$, $3$ are cross-edges, cross-polygons and cross-polyhedra, respectively. The dimension of a CMC is the  maximum order of its cells.  The CMC of order $3$ depicted in  Fig. \ref{fig:CMC_1}(b)  consists of  $L=3$ layers.  Specifically, it is composed of three intra-layer cell complexes, denoted by $\mathcal{X}^{1}, \mathcal{X}^{2}$ and $\mathcal{X}^{3}$, interconnected by cross-edges shown as dashed lines. We can observe  three cross-cells of order $2$, a  triangle, a square and  a pentagon, between layers $1$ and $2$, and  a  tetrahedron, between layers $2$ and $3$.\\
Let us consider the network layers indexed according to an (arbitrary) increasing order. For simplicity of notation and w.l.o.g., let us assume that cross-cells involve only a couple of layers,   as in the example illustrated in Fig. \ref{fig:CMC_1}(b). Hence, we denote by  $c_{k}^{\ell,m}(i)$  the $i$-th cross-cell of order  $k > 0$, interconnecting layers $\ell$ and $m$.
 Furthermore, we denote by $c_{k}^{\ell}$ the intra-layer cells, i.e. cells of order $k$ within layer $\ell$. 

Given the cross-cell $c_{k}^{\ell,m}(i)$, we define its $\ell$-layer faces and $m$-layer faces  as the cells of order $0 \leq j<k$ that lower bound $c_{k}^{\ell,m}(i)$ and belong to the $\ell$-layer and $m$-layer, respectively.


  In the  example shown  in Fig. \ref{fig:CMC_1}(b), the 2-order cross-cell $c_2^{1,2}(1)$   is a cross-triangle, connecting layers $1$ and $2$,  with bounding $1$-order cells $c_{1}^{1}(1),c_{1}^{1,2}(1), c_{1}^{1,2}(2)$. The  face  of $c_2^{1,2}(1)$ on layer $1$ is the edge $c_1^{1}(1)$, while its face on layer $2$ is the node $c_0^{2}(1)$. Note that, in general, cross-cells may have faces of different orders on each layer.
  
 Therefore, a cell multicomplex $\mathcal{X}$ is defined as  a collection of intra- and cross-layer complexes.
We denote by $\mathcal{X}^{\ell}$ the  cell complex within layer $\ell$ and by $\mathcal{X}^{\ell,m}$ the cross-complex  composed of the cross-cells inter-connecting layers $\ell$ and $m$.  Furthermore, we define the cross-complex $\mathcal{X}_{k,n}^{\ell,m} \subseteq \mathcal{X}^{\ell,m}$ as the collection of cross-cells with faces of order $k$ in layer $\ell$ and with faces of order $n$ in layer $m$. Hence, we denote by $N_{k,n}^{\ell,m}$ the number of cross-cells in $\mathcal{X}_{k,n}^{\ell,m}$, i.e. $N_{k,n}^{\ell,m}=|\mathcal{X}_{k,n}^{\ell,m}|$. Extending this notation, we define  with $\mathcal{X}_{k,-1}^{\ell,m}$    the   intra-layer cell complex of order $k$ in layer $\ell$, where the subscript $-1$ indicates no cells over layer $m$. Then, it holds that $\mathcal{X}^{\ell} \equiv \mathcal{X}_{K,-1}^{\ell,m}$ and $\mathcal{X}^{m} \equiv \mathcal{X}_{-1,K}^{\ell,m}$ where $K$ is the maximum dimension of the cells in the intra-layer complex.

\begin{figure}[t]
\centering
\includegraphics[width=8.5cm,height=7.2cm]{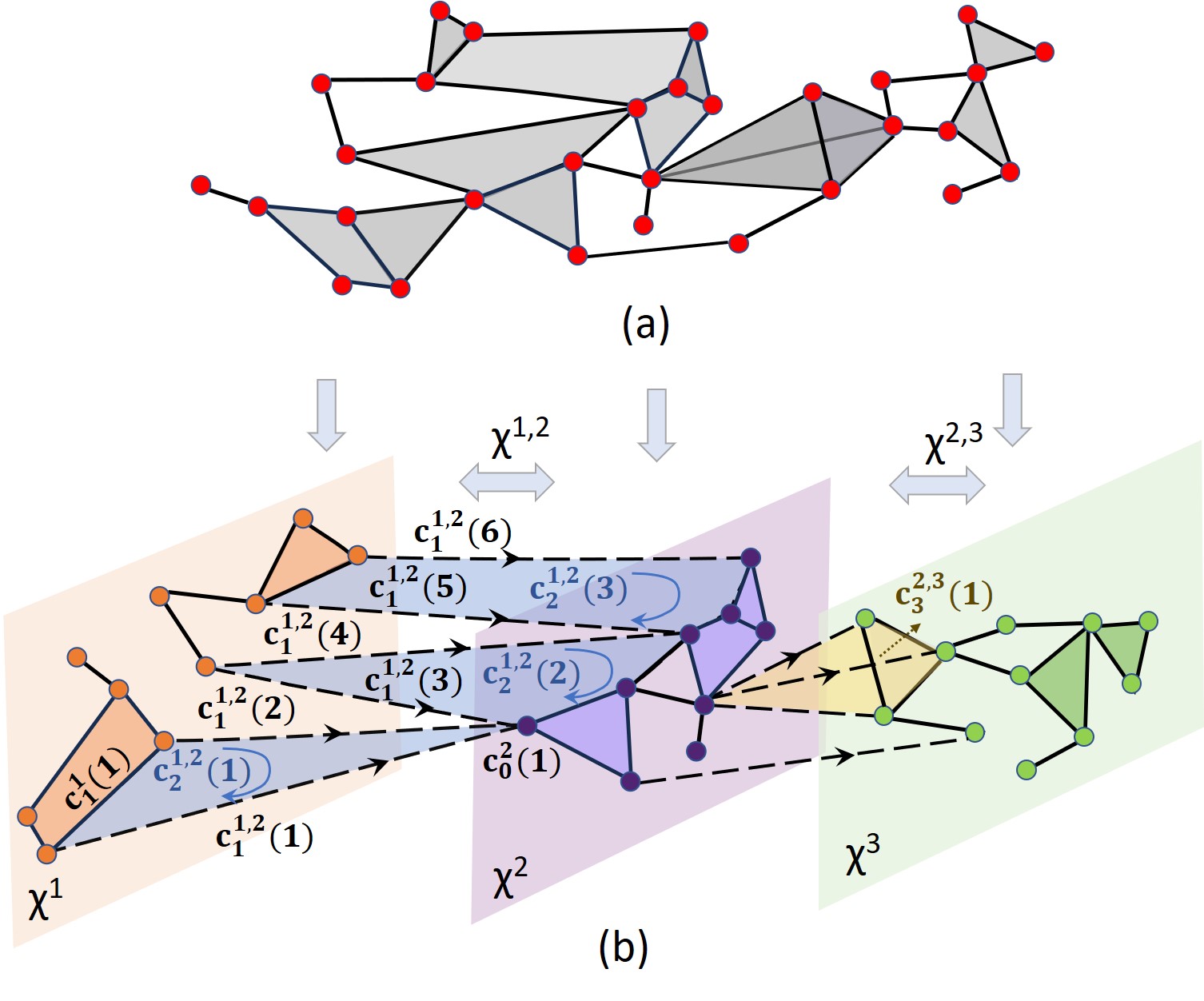}
\caption{An example of (a) an ACC of order $3$ and  (b) its CMC representation  with $L=3$ layers.}
\label{fig:CMC_1}
\end{figure}
Considering the CMC illustrated in  Fig. \ref{fig:CMC_1}(b) the cross-complex $\mathcal{X}^{1,2}$ between layers $1$ and $2$, is given by $\mathcal{X}^{1,2}=\{\mathcal{X}^{1,2}_{0,0},\mathcal{X}^{1,2}_{1,0},\mathcal{X}^{1,2}_{0,1},\mathcal{X}^{1,2}_{1,1}\}$ with  $\mathcal{X}^{1,2}_{0,0}= \{c_{1}^{1,2}(i)\}_{i=1}^{6}$, $\mathcal{X}^{1,2}_{1,0}=\{c_2^{1,2}(1)\}, \mathcal{X}^{1,2}_{0,1}=\{c_2^{1,2}(2)\}$,  and, finally, 
 $\mathcal{X}^{1,2}_{1,1}= \{c_2^{1,2}(3)\}$. Note that the complex $\mathcal{X}_{0,0}^{1,2}$ is the cross-graph, i.e. a cell complex of order $1$, while $\mathcal{X}^{1,2}_{1,0}$, $\mathcal{X}^{1,2}_{0,1}$ and $\mathcal{X}^{1,2}_{1,1}$ are cross-complexes of order $2$.
\\
\textbf{Orientation of CMCs.} It is useful to introduce the orientation of the cells, a choice made only for the algebraic representation of the CMC. The orientation of a (cross-)cell is defined by choosing an ordering of its lower-dimensional  bounding cells and can be derived by generalizing the notion of simplex orientation for SC \cite{grady2010}.
Every simplex admits only two possible orientations, depending on the ordering of its elements. Two orderings represent the same orientation if they differ by an even number of transpositions, where a transposition is a  permutation of two elements \cite{munkres2000topology}.
To define the orientation of a $k$-order cell, we may apply a simplicial decomposition, 
which consists in partitioning the cell into a set of internal $k$-simplices \cite{grady2010}, \cite{sardellitti2024topological}.
We use the notation $c_{k-1}^{\ell,m}(i)  \sim c_{k}^{\ell,m}(j)$ to indicate that the orientation of $c_{k-1}^{\ell,m}(i)$ is coherent with that of  $c_{k}^{\ell,m}(j)$ and $c_{k-1}^{\ell,m}(i) \nsim c_{k}^{\ell,m}(j)$ to indicate that their orientations are opposite. An oriented $k$-order (cross-)cell $c_{k}^{\ell,m}$ of a CMC may be represented through its bounding $(k-1)$-cells  with two consecutive  cells sharing a common  $(k-2)$-cell. 
Given an orientation, there are two ways in which two cross-cells can be considered to be adjacent: lower and upper adjacent.
Two $k$-order (cross-)cells are lower adjacent if they share a common face of order $k-1$ and upper adjacent if they are both faces of a cell of order $k+1$.



 \section{Algebraic framework for CMCs: from global to local homologies}
\label{sec: Alg_rep}
In many emergent applications, from data science to machine learning, data resides on multiple interconnected domains or networks, and 
the objective may be uncover  global as well as local topological features that characterize the structure and relationships within and between these domains. 

  Depending on the learning task, we can adopt two main approaches  for the analysis of  signals over cell multicomplexes. 
  In the first one, the topological structure is treated as a monolayer  domain, so that the Hodge-Laplacian  introduced for representing cell complexes can be used to process  signals as in \cite{sardellitti2024topological}. In the second, novel approach,  we leverage cross-Laplacian matrices for signal representation in order to capture intra- and inter-layer homologies and reveal local topological invariants within the cell complex.
 
 \subsection{Cell multicomplexes: the monocomplex perspective}
 
 One of the common approaches for the analysis of multilayer networks   is to represent them as   monocomplex structures 
  to take advantage of well-established single-layer algebraic representation methods \cite{kivela2014multilayer},\cite{bianconi2018multilayer}.
 This allows to capture the topological invariants of the   entire network by treating it as a unique entity.

 Therefore,  the flattened multicomplex can be algebraically  represented  using  the  Hodge-Laplacian matrix \cite{sardellitti2024topological}. In the following, w.l.o.g., we focus on a $2$-order cell multi-complex $\mathcal{X}$.
Let us assume that the CMC consists of $L$ interconnected  layers, where the layer indices are ordered increasingly as  $1,2,...,L$. 
Denote by  $\mathcal{G}^{\ell}=(\mathcal{V}_{\ell},\mathcal{E}_{\ell})$ the intra-layer graph 
where  $\mathcal{V}_{\ell}$ and $\mathcal{E}_{\ell}$ are the sets of $N_{\ell}$ nodes and $E_{\ell}$ edges on layer ${\ell}$, respectively.
 To simply our notation we define the set $\mathcal{X}^{\ell,m}_{0,0}$ of cross-edges  connecting layer $\ell$ and $m$  as $\mathcal{X}^{\ell,m}_{0,0}=\mathcal{E}_{\ell,m}$  with $|\mathcal{E}_{\ell,m}|=E_{\ell,m}$.  Then, we can consider the CMC $\mathcal{X}$ as a single cell complex whose underlying graph is defined as $\mathcal{G}=(\mathcal{V},\mathcal{E})$, with $\mathcal{V}$ and $\mathcal{E}$, the set of nodes and edges in $\mathcal{X}$, respectively. Specifically, we have
 a total number of nodes  $|\mathcal{V}|=N$ and of  edges $|\mathcal{E}|=E$ with $N=\sum_{l=1}^{L} N_{l}$ and 
 $E=\sum_{l=1}^{L} E_{l}+\sum_{l=1}^{L}\sum_{m=1, m>l}^{L} E_{l,m}$.
 In order to introduce the incidence matrix $\mB_k$ describing which $k$-cells are upper adjacent to which $(k-1)-$cells, we define w.l.o.g.  an ordering of the cells in the complex. 
Specifically, we assume that the row indices of $\mB_1$ are associated in sequence with the nodes of each intra-layer graph $\mathcal{G}^{\ell}$, and the column indices are ordered as 
$E_1,E_{1,2},\ldots, E_{1,L},E_2,E_{2,3},\ldots,E_{L-1},E_{L-1,L}, E_L$.
Denoting with $P_i$ and $P_{i,j}$ for $i,j=1,\ldots,L$  the number of $2$-cells over layer $i$ and between layers $i$ and $j$, respectively, the $2$-order cells are ordered as $P_1, P_{1,2},\ldots, P_{1,L},P_2,P_{2,3},\ldots,P_{2,L},\ldots,P_{L-1},P_{L-1,L}, P_L$.\\
We assume w.l.o.g. that cross-edges are oriented from layer $\ell-1$ to layer $\ell$.
Therefore, the incidence (or boundary) matrix $\mB_k$ for $k=1,2$
can be defined as in \cite{sardellitti2024topological}:
\beq \label{B_1_mono}
  B_{k}(i,j)=\left\{\begin{array}{rll}
  0, & \text{if} \; c_{k-1}(i) \not\prec_b c_{k}(j) \\
  1,& \text{if} \; c_{k-1}(i) \prec_b c_{k}(j),    \; c_{k-1}(i) \sim c_{k}(j)\\
  -1,& \text{if} \; c_{k-1}(i) \prec_b c_{k}(j),    \; c_{k-1}(i) \nsim c_{k}(j)\\
  \end{array}.\right.
  \eeq 
  An important property of the boundary matrices is that the boundary of a boundary is zero, i.e. $\mB_{k} \mB_{k+1}=\mathbf{0}$ \cite{Lim}. 
Then, we can build the vertex-edge incidence matrix $\mB_1$, describing the lower incidences of the edges, and the  edge-polygon incidence matrix $\mB_2$, describing the upper incidences of the edges. Finally, we can represent the cell multicomplex $\mathcal{X}$ through the Hodge Laplacian matrices \cite{goldberg2002combinatorial}
\beq
\begin{array}{ccc}
&\mL_0=\mB_1\mB_1^T,  \qquad
& \mL_1=\mB_1^T\mB_1+\mB_2\mB_2^T
\end{array}
\eeq
where $\mL_0$ is the graph Laplacian, while $\mL_1$ is the first-order Laplacian matrix of the cell complex \cite{sardellitti2024topological}.
Specifically, the two orthogonal matrices $\mL_{1,d}=\mB_1^T \mB_1$ and  $\mL_{1,u}=\mB_{2} \mB_{2}^T$ are the lower and upper Laplacians, respectively, since they express the lower and upper adjacencies of the edges. 

One of the key properties of the Hodge Laplacian matrix $\mL_1$ is that it induces the so-called Hodge decomposition of the edge space $\mathbb{R}^{E}$ into three orthogonal subspaces as \cite{Lim}
\beq \label{eq: Hodge_all}
\mathbb{R}^{E} \equiv \text{img}(\mB_1^T) \oplus \text{ker}(\mL_1) \oplus \text{img}(\mB_2)
\eeq
where $\text{ker}(\mL_1)$  contains the vectors in both the $\text{ker}(\mB_1)$ and $\text{ker}(\mB_2^T)$.\\
\textbf{Betti numbers.}   The Hodge Laplacian representation of a CMC captures global invariants of the underlying topological spaces. 
The kernel of $\mL_1$
defines the homology group (invariants) of $\mathcal{X}$, which algebraically encodes the presence of holes or cycles that cannot be continuously deformed into boundaries \cite{Lim}. In particular, the dimension of the kernel of the $k$-order Hodge Laplacian, given by $\beta_{k}=\text{dim}(\text{ker}(\mL_k))$, is called the $k$-th Betti number. These numbers count the number of $k$-dimensional holes in the complex.
Specifically,  $\beta_{0}$  represents the number of connected components of the multilayer graph;    $\beta_{1}$ indicates  the number of $1$-dimensional holes in the entire complex, i.e. the number of empty $2$-cells within the complex; $\beta_2$ represents the number of cavities and so on \cite{hatcher2005algebraic}.\\
\textbf{Incidence matrix construction.} As an illustrative example, let us derive the incidence matrices for a network with $L=2$ layers. The node-edge incidence matrix $\mB_1$ can be written as:
\beq \label{eq:mono_B1}
\mB_1=\left[ \begin{array}{lll}
\mB_{1}^{(1)} & \mB_{1}^{1,(2)} & \mathbf{0}\\
\mathbf{0} & \mB_{1}^{(1),2} & \mB_{1}^{(2)} 
\end{array}\right]
\eeq
 where:
     i) $\mB_1^{(1)}$, $\mB_1^{(2)}$ are the node-edge incidence matrices of the graphs in layer $1$ and layer $2$, respectively;
     ii) $\mB_1^{1,(2)}$ (or $\mB_1^{(1),2}$)  has entries $1$ or $-1$, according to the chosen edge orientation, on the rows  corresponding to the vertices of layer $1$ (or $2$)  that are endpoints of the cross-edges between layers $1$ and $2$.  
The edge-polygon matrix $\mB_2$  is expressed as 
\beq \label{eq:mono_B2}
\mB_2=\left[ \begin{array}{lll}
\mB_{2}^{(1)} & \mB_{2}^{1,(2)} & \mathbf{0}\\
\mathbf{0} & \mB_{2}^{1,2} & \mathbf{0}\\
\mathbf{0} & \mB_{2}^{(1),2} & \mB_{2}^{(2)}
\end{array}\right]
\eeq
 where:
     i) $\mB_2^{(1)}$, $\mB_2^{(2)}$ are the edge-polygon incidence matrices of layer $1$ and layer $2$, respectively;
     ii) $\mB_2^{1,(2)}$ (or $\mB_1^{(1),2}$)  has entries $1$ or $-1$, according to the chosen edge orientation, in correspondence of the edges on layer $1$ (or $2$) bounding the  $2$-order cross-cells;
   iii) $\mB_2^{1,2}$ has entries $1$ or $-1$, according to the chosen edge orientation, in correspondence of the cross-edges between layers $1$ and $2$ bounding  $2$-order cross-cells.

\vspace{-0.3cm}
 \subsection{Cross-Laplacians to capture cross-invariants}
Although representing a cell multicomplex as a single complex is useful in contexts where global invariants of the space are analyzed, this representation often fails to grasp local features from data. Then,
representing data as a  monocomplex topological structure can lead to a loss of key topological information. 
 In many cases, depending on the learning task, is required a topological representation  that can disentangle the local and global homologies of the layers and reveal how the topology of one layer
affects and controls the topology of others. 
This perspective allows to see each layer as exhibiting   different topological properties depending on how it is explored: whether we look at each layer from its point of view, through the lens of other layers, or as a part of a whole aggregate structure.

In this section, we introduce the notion of cross-Laplacian matrices, originally presented in \cite{moutuou2023} for simplicial complexes, and extend it to the settings of cell multi-complexes.\\
\textbf{Cross-boundaries maps.}
  Let us first introduce the notion of boundary maps of cross-cells in the perspective of  a specific layer, i.e. the boundary maps of cross-cells only respect to faces  belonging to a given layer and keeping fixed all the remaining faces.
  
  Let us consider  two layers $\ell,m$ and denote by $C_{k,n}$ the real vector space generated by  all oriented $q$-order $(k,n)-$cross-cells $c_q^{\ell,m}$, with faces of order $k$ on layer $\ell$ and faces of order $n$  on layer $m$.
  \begin{table}[ht]
  \centering
   \caption{Orders $(k,n)$ of the faces on the layers of cross-cells of order $q\leq 3$.}
  \begin{tabular}{l|l}
    \hline
    
    Cell order $q$ & $(k,n)$   \\
    
    \hline
    $1$         &  $(0,0)$                  \\
    $2$         &  $$(1,0), (0,1), (1,1)$$      \\
    $3$         &  $$(1,1), (1,2) (2,1), (2,2)$$   \\
    \hline
  \end{tabular}
 
  \label{tab:q_cell_orders}
\end{table}  
  To simplify our notation, we omit the dependence of the cell's order $q$  on the orders $(k,n)$ of the faces on  layers $\ell$ and $m$, respectively. Considering a $2$-order CMC, for cross-edges we have $q=1$ and  $(k,n)=(0,0)$. For $q=2$, i.e. considering  $2$-order cross-cells, we can have $(k,n)=(1,0),(0,1),(1,1)$ and so on, as illustrated in the Table  \ref{tab:q_cell_orders}. Note that according to our notation, a pair $(k,n)$ may also correspond to cross-cells of different orders.

  Hence, given the cross-complex  $\mathcal{X}_{k,n}^{\ell,m}$  we can define two distinct cross-boundary operators for each cross-cell $c_{q}^{\ell,m}\in \mathcal{X}_{k,n}^{\ell,m}$, as each cross-cell can be viewed from two different perspectives: either from layer $\ell$ or from layer $m$.\\  
The first operator $\mB_{k,n}^{(\ell),m}$ is a boundary map defined with respect to the faces on layer $\ell$, while the second operator, denoted as $\mB_{k,n}^{\ell,(m)}$, is a boundary map with respect to the faces on layer $m$. We adopt a notation for the boundary operators that simplifies the construction of the $(0,0)$ cross-Laplacians introduced below, noting that, for further generalizations, the notation  must be tailored to the $(k,n)$ boundaries under study and to the perspective from which we analyze them. 
Here, we express the boundary map with respect to the cells of order $k$ on layer $\ell$ as view from layer $m$ as
$ \mB_{k,n}^{(\ell),m}: C_{k,n} \rightarrow C_{k-1,n}$.
Therefore, $\mB_{k,n}^{(\ell),m}$ is a boundary operator that maps cells $c_{q}^{\ell,m} \in \mathcal{X}_{k,n}^{\ell,m}$ to cells $c_{q-1}^{\ell,m} \in \mathcal{X}_{k-1,n}^{\ell,m}$. It is an incidence  matrix of dimension $N_{k-1,n}^{\ell,m} \times N_{k,n}^{\ell,m}$ with entries defined as 
    \beq \label{eq:B_kn_ell}
  B_{k,n}^{(\ell),m}(i,j)=\!\!\left\{\!\!\!\begin{array}{rll}
  0, & \! \text{if} \; c_{q-1}^{\ell,m}(i) \not\prec_b c_{q}^{\ell,m}(j) \medskip\\
  1,& \!\text{if} \; c_{q-1}^{\ell,m}(i) \prec_b c_{q}^{\ell,m}(j) , \; c_{q-1}^{\ell,m}(i) \sim c_{q}^{\ell,m}(j) \medskip\\
  -1,& \!\text{if} \; c_{q-1}^{\ell,m}(i) \prec_b c_{q}^{\ell,m}(j), \;   c_{q-1}^{\ell,m}(i) \nsim c_{q}^{\ell,m}(j)\\
  \end{array}\right. 
  \eeq
  where $c_{q}^{\ell,m}(i)\in \mathcal{X}_{k,n}^{\ell,m}$ and $c_{q-1}^{\ell,m}(j)\in \mathcal{X}_{k-1,n}^{\ell,m}$, $\forall i,j$. \\
   Similarly, we can define the boundaries matrices $\mB_{k,n}^{\ell,(m)}: C_{k,n} \rightarrow C_{k,n-1}$ of dimension 
   $N_{k,n-1}^{\ell,m} \times N_{k,n}^{\ell,m}$ with respect to faces on layer $m$, as 
    \beq \label{eq:B_kn_m}
  B_{k,n}^{\ell,(m)}(i,j)=\!\!\left\{\!\!\!\begin{array}{rll}
  0, & \! \text{if} \; c_{q-1}^{\ell,m}(i) \not\prec_b c_{q}^{\ell,m}(j)  \medskip\\
  1,& \! \text{if} \; c_{q-1}^{\ell,m}(i) \prec_b c_{q}^{\ell,m}(j), \; c_{q-1}^{\ell,m}(i) \sim c_{q}^{\ell,m}(j) \medskip\\
  -1,& \! \text{if} \; c_{q-1}^{\ell,m}(i) \prec_b c_{q}^{\ell,m}(j), \; c_{q-1}^{\ell,m}(i) \nsim c_{q}^{\ell,m}(j)\\
  \end{array}\right. 
  \eeq
  where $c_{q}^{\ell,m}(i)\in \mathcal{X}_{k,n}^{\ell,m}$ and $c_{q-1}^{\ell,m}(j)\in \mathcal{X}_{k,n-1}^{\ell,m}$, $\forall i,j$. 
  For consistency in our notation,  we denote  the intra-layer cell complexes on layer $\ell$ and $m$, by $\mathcal{X}_{k,-1}^{\ell, m}$ and $\mathcal{X}_{-1,n}^{\ell, m}$, respectively. Additionally, note that it holds 
\begin{align}
&\mB_{k,-1}^{\ell,(m)}=\mathbf{0}, \;\; 0\leq k \leq K \label{eq:B_k_null} \\ & \mB_{-1,n}^{(\ell),m}=\mathbf{0}, \;\;  0 \leq n \leq K \label{eq:B_k_null1} \\
& \mB_{0,-1}^{(\ell),m}=\mathbf{0},\; \;  \mB_{-1,0}^{\ell,(m)}=\mathbf{0} \label{eq:B_k_null2} 
\end{align}
where $K$ is the dimension of the CMC $\mathcal{X}$.
  The  key property stated in the following proposition  represents  the foundation of the homological structure of the CMCs we are introducing.
  \begin{proposition}
  Consider a CMC $\mathcal{X}$ composed by a set of (cross-)complexes $\mathcal{X}_{k,n}^{\ell, m} \subset \mathcal{X}$. Then, the following orthogonality conditions hold:
      \beq \label{eq:orth}
      \begin{split}
  &\text{i)} \; \mB_{k,n}^{(\ell),m}\mB_{k+1,n}^{(\ell),m}=\mathbf{0} \quad\\
  &\text{ii)}\,  \mB_{k,n}^{\ell,(m)}\mB_{k,n+1}^{\ell,(m)}=\mathbf{0}. 
  \end{split}\eeq 
    \end{proposition}
    \begin{proof}
    See Appendix A.
    \end{proof}

\begin{figure}[t]
\centering
\includegraphics[width=8cm,height=5.0cm]{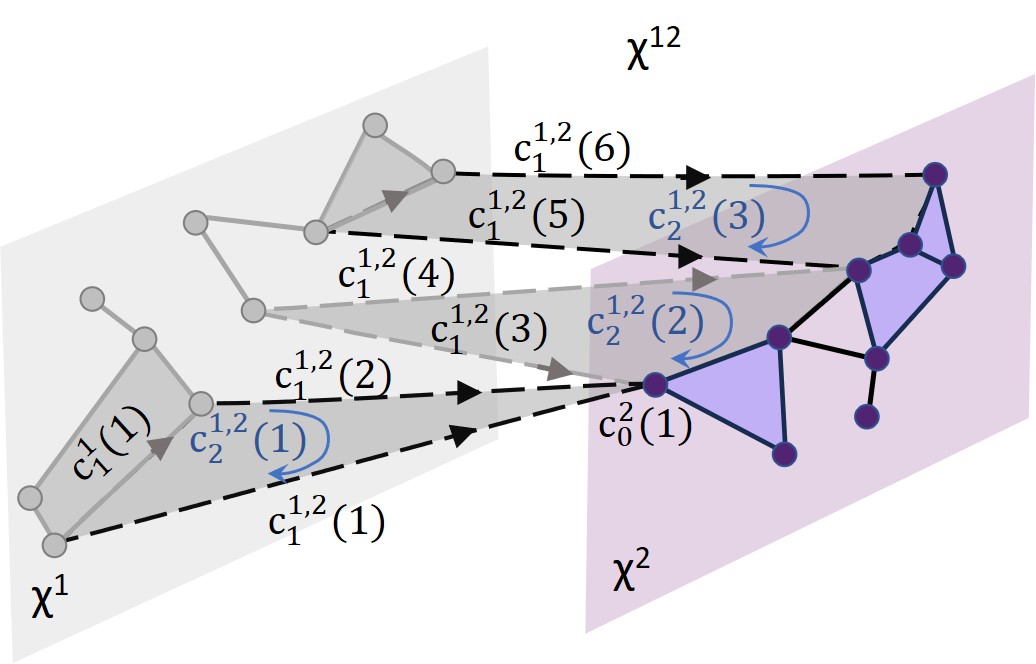}
\caption{Cross-complex $\mathcal{X}^{1,2}$ as viewed from layer $2$.}
\label{fig:CMC_local}
\end{figure}
As an example, let us consider the simple cross-cell complex $\mathcal{X}^{1,2}_{1,0}=\{ c_{2}^{1,2}(1)\}$ in Fig. \ref{fig:CMC_local}. The $2$-order cross-cell $c_{2}^{1,2}(1)$   has one face of order $1$ (edge) on layer $1$ and  one face of order $0$ (vertex) on layer $2$. The bounding cells of $c_{2}^{1,2}(1)$   are: with respect to cells on layer $1$ the two cross-edges $c_{1}^{1,2}(1)$ and $c_{1}^{1,2}(2)$, while with respect to cells on layer $2$ the bounding cell is $c_1^{1}(1)$.  
Therefore,  the   matrix $\mB_{1,0}^{(1),2} \in \mathbb{R}^{N_{0,0}^{1,2}\times N_{1,0}^{1,2}}$,  defined according to (\ref{eq:B_kn_ell}) with $N_{0,0}^{1,2}=6$ and $N_{1,0}^{1,2}=1$, is given by
$\mB_{1,0}^{(1),2}=[-1,1,0,0,0,0]^T$. 
Similarly, the matrix $\mB_{1,0}^{1,(2)} \in \mathbb{R}^{N_{1,-1}^{1,2} \times N_{1,0}^{1,2}}$ with $N_{1,-1}^{1,2}=10$ (number of edges on layer $1$) and $ N_{1,0}^{1,2}=1$, is
$\mB_{1,0}^{1,(2)}=[1,0,0,0,0,0,0,0,0,0]^T$.
It is useful to remark that these matrices can be extracted by the monocomplex incidence matrix $\mB_2$ in (\ref{eq:mono_B2}). Specifically, $\mB_{1,0}^{(1),2}$ and $\mB_{1,0}^{1,(2)}$ are the columns of $\mB_{2}^{1,2}$ and $\mB_{2}^{1,(2)}$, respectively, corresponding to the $(1,0)$ cell $c_{2}^{1,2}(1)$. \\
\textbf{Cross-Laplacian matrices.}  
Given the two layers $\ell,m$, we  are now able to introduce  a set of topological descriptors called the  $(k,n)$-cross-Laplacian matrices and describing local homologies from the perspective of a given layer. Then, we define the $(k,n)$-cross-Laplacian matrix   from layer\footnote{To maintain clarity in the notation, we define the cross-Laplacian from the layer with respect to which the boundary is calculated, while noting that the space is viewed from the perspective of the other layer.}  
$\ell$ as
\beq \label{eq:Lknl_l}
\mL_{k,n}^{(\ell),m}=\underbrace{{(\mB_{k,n}^{(\ell),m})^{T}\mB_{k,n}^{(\ell),m}}}_{\text{Lower Laplacian from layer $\ell$}} +\underbrace{\mB_{k+1,n}^{(\ell),m} (\mB_{k+1,n}^{(\ell),m})^{T}}_{\text{Upper Laplacian from layer $\ell$}}  
\eeq
where the first and second term encode the lower and upper adjacencies of cross-cells belonging to $\mathcal{X}_{k,n}^{\ell, m}$, respectively. Similarly, the  $(k,n)$-cross-Laplacian matrices from layer $m$ is defined as 
\beq \label{eq:Lknl_m}
\mL_{k,n}^{\ell,(m)}=\underbrace{(\mB_{k,n}^{\ell,(m)})^{T}\mB_{k,n}^{\ell,(m)}}_{\text{Lower Laplacian from layer $m$}}  +\underbrace{\mB_{k,n+1}^{\ell,(m)}(\mB_{k,n+1}^{\ell,(m)})^T}_{\text{Upper Laplacian from layer $m$}}.
\eeq
These Laplacian matrices are symmetric and semidefinite positive. 
It can be observed that the layer $\ell$ Hodge Laplacian of order $k$  can be derived from (\ref{eq:Lknl_l}) by setting $n=-1$, as
\beq \label{eq:Lk-1}
\mL_{k,-1}^{(\ell),m}= (\mB_{k,-1}^{(\ell),m})^{T}\mB_{k,-1}^{(\ell),m}+\mB_{k+1,-1}^{(\ell),m}(\mB_{k+1,-1}^{(\ell),m})^T. 
\eeq
Similarly, the layer $m$ Hodge Laplacian of order $n$ is obtained from (\ref{eq:Lknl_m})  by setting $k=-1$.\\
\textbf{Hodge decompositions of the space.}
 Following similar considerations as in \cite{Lim,moutuou2023} and using (\ref{eq:orth}), we  can prove that   the space $\mathbb{R}^{N_{k,n}^{\ell,m}}$  
admits different Hodge decompositions, depending on the layer with respect to which the boundaries are calculated. Then, we get
\beq \label{eq: Hodge_space_dec_ell}
\mathbb{R}^{N_{k,n}^{\ell,m} }\!\equiv \text{img}(\mB_{k,n}^{(\ell),m \,T}) \oplus \text{ker}(\mL_{k,n}^{(\ell),m}) \oplus  \text{img}(\mB_{k+1,n}^{(\ell),m}), 
\eeq
\beq \label{eq: Hodge_space_dec_m}
\mathbb{R}^{N_{k,n}^{\ell,m} }\!\equiv \text{img}(\mB_{k,n}^{\ell,(m) \, T}) \oplus \text{ker}(\mL_{k,n}^{\ell,(m)}) \oplus  \text{img}(\mB_{k,n+1}^{\ell,(m)}).       
\eeq
The orthogonality conditions in (\ref{eq:orth}) allow to define the $\ell$ and $m$ layer $(k,n)$-cross-homology groups of $\mathcal{X}$ \cite{Lim,moutuou2023} as  \beq 
\label{eq:Hkn} \text{H}_{k,n}^{(\ell)}\cong \text{ker} (\mL_{k,n}^{(\ell),m}), \;  \; \text{H}_{k,n}^{(m)}\cong \text{ker} (\mL_{k,n}^{\ell,(m)}). \eeq 
The cross-homology groups are characterized by their dimensions, named the $\ell$ and $m$ $(k,n)$-cross-Betti numbers $\beta_{k,n}^{(\ell)}= \text{dim} (\text{H}_{k,n}^{(\ell)})$ and $\beta_{k,n}^{(m)}= \text{dim} (\text{H}_{k,n}^{(m)})$. Then, we can define the $(k,n)$-cross-Betti vector of $\mathcal{X}_{k,n}^{\ell,m}$ as 
\beq \label{eq:Betti_vec}
\boldsymbol{\beta}_{k,n}^{\ell,m}=[\beta_{k,n}^{(\ell)},\beta_{k,n}^{(m)}].
\eeq
As we will discuss below, these numbers  are able to capture the homologies of the intra-layer and cross-layer cell complexes, by identifying different invariants of the spaces  depending on the indexes $(k,n)$.

\section{Multilayer graphs  through the lens of cross-Laplacians}
\label{sec: Multilayer_graphs}
In this section we explore how multilayer graphs, i.e. CMCs of order $1$, can be represented using  cross-Laplacians. 
We construct the node-edge boundary operators acting on  subsets  of vertices as well as  on intra- and cross-edges. \\
The $\ell$ (intra-)layer $(0,-1)$-cross-Laplacian $\mL_{0,-1}^{(\ell),m}$ represents the common $\ell$ layer graph Laplacian of order $0$. It can be obtained from (\ref{eq:Lk-1}),  setting  
$k=0$ and using  (\ref{eq:B_k_null2}), as
\beq
\mL_{0,-1}^{(\ell),m}=\mB_{1,-1}^{(\ell),m} (\mB_{1,-1}^{(\ell),m})^{T}
\eeq
where  the node-edge incidence matrix $\mB_{1,-1}^{(\ell),m}$ is derived from (\ref{eq:B_kn_ell}) with $q=1$. 
 The cross-Laplacian matrix $\mL_{0,-1}^{\ell,(m)}$ from layer $m$  is derived    using  (\ref{eq:Lknl_m}) and (\ref{eq:B_k_null}) as
  \beq
\mL_{0,-1}^{\ell,(m)}=\mB_{0,0}^{\ell,(m)} (\mB_{0,0}^{\ell,(m)})^T
\eeq
where $\mB_{0,0}^{\ell,(m)}$ is the boundary matrix 
obtained from (\ref{eq:B_kn_m}) setting $q=0$ 
and with $c_{0}^{\ell}(j) \in \mathcal{X}^{\ell,m}_{0,-1}$ and $c_{1}^{\ell,m}(j) \in \mathcal{X}^{\ell,m}_{0,0}$.
  It can be easily shown that $\mL_{0,-1}^{\ell,(m)}$ is a diagonal matrix with diagonal entries the upper degrees of nodes on layer $\ell$ with respect to cross-edges. Thus, zeros on the diagonal  identify the nodes on layer $\ell$ that are not $\ell$ faces of cross-edges.
  \\
Similarly, from the perspective of layer  $m$, we get the common $m$ layer graph Laplacian of order $0$ as
$
\mL_{-1,0}^{\ell,(m)}=\mB_{-1,1}^{\ell,(m)} (\mB_{-1,1}^{\ell,(m)})^{T}
$
where $\mB_{-1,1}^{\ell,(m)}$ is the node-edge incidence matrix of the graph on layer $m$.
The cross-Laplacian matrix  from layer $\ell$ is the diagonal matrix
$
\mL_{-1,0}^{(\ell),m}=\mB_{0,0}^{(\ell),m} (\mB_{0,0}^{(\ell),m})^{T}
$
with diagonal entries representing the upper cross-degree of the nodes on layer $m$, i.e. the number of cross-edges to which each node is connected.\\
\textbf{Cross-Betti numbers to identify hub  nodes.}
Let us consider the  cross-Betti vector for a multi-layer graph where pairs of layers $(\ell,m)$ are interconnected. The Betti vector $\boldsymbol{\beta}_{0,-1}^{\ell,m}$ is defined from (\ref{eq:Betti_vec}) as
\beq
\boldsymbol{\beta}_{0,-1}^{\ell,m}=[\beta_{0,-1}^{(\ell)},\beta_{0,-1}^{(m)}]
\eeq
where: i) $\beta_{0,-1}^{(\ell)}=\text{dim}(\text{ker}(\mL_{0,-1}^{(\ell),m}))$ is the number of connected components in the graph $\mathcal{G}^{\ell}$, ii) $\beta_{0,-1}^{(m)}=\text{dim}(\text{ker}(\mL_{0,-1}^{\ell,(m)}))$ is the number of nodes in $\mathcal{G}^{\ell}$ that are not connected with any nodes in $\mathcal{G}^{m}$.
Similarly considerations hold for the cross-Betti vector $\boldsymbol{\beta}_{-1,0}^{\ell,m}$.\\
Interestingly, since  $\mL_{0,-1}^{\ell,(m)}$ and $\mL_{-1,0}^{(\ell),m}$
are diagonal matrices their eigenvectors identify the hub nodes
over each layer that play a key role for the interconnections between layers.
\textbf{Remark.} Note that   the above boundary matrices can be extracted from the monocomplex  node-edge incidence matrix in (\ref{eq:mono_B1}). Specifically, for a two-layer graph, $\mB_{1,-1}^{(1),2}$, $\mB_{0,0}^{1,(2)}$, $\mB_{-1,1}^{1,(2)}$, $\mB_{0,0}^{(1),2}$ correspond to the sub-matrices $\mB_{1}^{(1)}$, $\mB_{1}^{1,(2)}$, $\mB_{1}^{(2)}$ and $\mB_{1}^{(1),2}$, respectively.

\section{Cross-Laplacian representation of CMCs in the edge space}
\label{sec: CMC_in_edge_space}
One of the most appealing features of the framework we here propose is its  capability to represent different levels of connectivity using a plethora of algebraic descriptors of the observed data. Considering a second order CMC, we can build various cross-Laplacians  depending on the invariant of the space we aim to capture.\\
In this work, we focus on the $(0,0)$-cross-Laplacians, i.e. the Laplacian matrices defined over the cross-edge space. We show how these matrices encode  local invariant of the topological space through the cross-Betti vector.\\
\textbf{Cross-Laplacians in the cross-edge space.}
Let us consider the $(0,0)$ cross-Laplacians $\mL_{0 ,0}^{(\ell), m}$ and $\mL_{0,0}^{\ell, (m)}$. These Laplacians  are $N_{0,0}^{\ell,m}\times N_{0,0}^{\ell,m}$-symmetric matrices indexed on the cross-edges $c_1^{\ell,m} \in \mathcal{X}_{0,0}^{\ell,m}$.
\\In particular,   from (\ref{eq:Lknl_l}), we get the Laplacian $\mL_{0 ,0}^{(\ell), m}$ as
\beq \label{ew_L_00_ell}
\mL_{0,0}^{(\ell),m}= (\mB_{0,0}^{(\ell),m})^{T} \mB_{0,0}^{(\ell),m}+\mB_{1,0}^{(\ell),m}(\mB_{1,0}^{(\ell),m})^{T}\eeq
where  the matrix  $\mB_{0,0}^{(\ell),m}$ of dimension $N_{-1,0}^{\ell,m} \times N_{0,0}^{\ell,m}$ is derived from  (\ref{eq:B_kn_ell}), as
  \beq \label{B_00l}
  B_{0,0}^{(\ell),m}(i,j)=\left\{\begin{array}{rll}
  0, &  \text{if} \; c_{0}^{m}(i) \not\prec_b c_{1}^{l,m}(j) \\
  1,& \text{if} \; c_{0}^{m}(i) \prec_b c_{1}^{l,m}(j),\;   \; c_{0}^{m}(i) \sim  c_{1}^{l,m}(j)\\
  -1,& \text{if} \; c_{0}^{m}(i) \prec_b c_{1}^{l,m}(j), \;   \; c_{0}^{m}(i) \nsim  c_{1}^{l,m}(j).\\
  \end{array}\right. 
  \eeq
Then, the lower Laplacian $\mathbf{{L}}_{d \, 0,0}^{(\ell),m}:= (\mB_{0,0}^{(\ell),m})^T\mB_{0,0}^{(\ell),m}$  captures the cross-edge lower incidences on layer $m$, i.e. the entry $i,j$ is equal to  $\pm 1$ if $c_{1}^{l,m}(i)$ is lower adjacent  to $c_{1}^{l,m}(j)$ on layer $m$, and $0$ otherwise.\\
Using (\ref{eq:B_kn_ell}),  we can derive the boundary matrix  
$\mB_{1,0}^{(\ell),m}: C_{1,0} \rightarrow C_{0,0}$ of dimension $N^{\ell,m}_{0,0} \times N^{\ell,m}_{1,0}$, with $N^{\ell,m}_{1,0} \in \mathcal{X}^{\ell,m}_{1,0}$ the number of $2$-order  $(1,0)$ cross-cells between layers $\ell,m$, as: 
\beq \label{B_00l}
  B_{1,0}^{(\ell),m}(i,j)\!=\!\!\left\{\begin{array}{rll}
  0, & \text{if} \; c_{1}^{\ell,m}(i) \not\prec_b c_{2}^{\ell,m}(j) \\
  1,& \text{if} \; c_{1}^{\ell,m}(i) \prec_b c_{2}^{\ell,m}(j),\,    c_{1}^{\ell,m}(i) \sim  c_{2}^{\ell,m}(j)\\
  -1,& \text{if} \; c_{1}^{\ell,m}(i) \prec_b c_{2}^{\ell,m}(j), \,   c_{1}^{\ell,m}(i) \nsim  c_{2}^{\ell,m}(j)\\
  \end{array}\right. 
  \eeq
  for $c_{1}^{\ell,m}(i) \in   \mathcal{X}_{0,0}^{\ell,m}$ and $c_{2}^{\ell,m}(j) \in \mathcal{X}_{1,0}^{\ell,m}$. The second term in (\ref{ew_L_00_ell})  is the upper Laplacian matrix $\mathbf{{L}}_{u \, 0,0}^{(\ell),m}:=\mB_{1,0}^{(\ell),m}(\mB_{1,0}^{(\ell),m})^{T}$ describing the upper adjacencies  of the cross-edges $c_1^{\ell, m}$ that are boundaries of $2$-order cross-cells with edges on layer $\ell$ and one vertex on layer $m$.\\ 
Similar considerations hold for the cross-Laplacian matrix  from layer $m$
\beq \label{ew_L_00_m}
\mL_{0,0}^{\ell,(m)}=(\mB_{0,0}^{\ell,(m)})^{T}\mB_{0,0}^{\ell,(m)}+\mB_{0,1}^{\ell,(m)} (\mB_{0,1}^{\ell,(m)})^{T}.
\eeq
It is worth noting that, considering a two-layer CMC, the matrices $\mB_{0,0}^{(1),2}$ and $\mB_{0,0}^{1,(2)}$
correspond to the sub-matrices $\mB_1^{(1),2}$ and  $\mB_1^{1,(2)}$  of the monocomplex matrix $\mB_1$ in (\ref{eq:mono_B1}), respectively; the matrices $\mB_{1,0}^{(1),2}$ and $\mB_{0,1}^{1,(2)}$ correspond instead to the columns of the sub-matrix $\mB_2^{1,2}$ in  (\ref{eq:mono_B2}) associated with $(1,0)$ and $(0,1)$ cross-cells, respectively.\\
\textbf{The Cross-Betti vector $\boldsymbol{\beta}_{0,0}$.}
To define the invariants of the topological spaces represented by the $(0,0)$-cross-Laplacians, we first introduce the notion of cones as defined in \cite{moutuou2023}. 
\begin{definition}
A cone is the shortest path of length two between two nodes within a  layer, passing  through a node on the other layer and not belonging to the cross-boundary of $2$-order cells.
A cone can be:
\begin{itemize}
    \item[i)] closed if it forms a cycle with intra-layer edges;
    \item[ii)] open if it has a vertex on a layer connecting two unconnected clusters on the other layer.
\end{itemize}
\end{definition}
Then, we can state the following theorem.
\begin{theorem}
    Given a CMC $\mathcal{X}$, the $(0,0)$ cross-Betti vector $\boldsymbol{\beta}_{0,0}^{\ell,m}=(\beta_{0,0}^{(\ell)},\beta_{0,0}^{(m)})$ counts the cones open and closed between layers $\ell$ and $m$.
    Specifically, $\beta_{0,0}^{(\ell)}$ counts the cones with vertices on layer $m$, while $\beta_{0,0}^{(m)}$ counts those with  vertices on layer $\ell$. These vertices are called harmonic cross-hubs.
    \end{theorem}
    \begin{proof}
    See Appendix B.
    \end{proof}

The harmonic cross-hubs are nodes on a layer  that control clusters on the other layer. The cross-hubs associated with open cones are key nodes for the connectivity among clusters, since, if removed, they might eliminate any communication between clusters of nodes.

Considering the cell multicomplex in Fig. \ref{fig:CMC_2}, 
we observe that $\boldsymbol{\beta}_{0,0}=(2,0)$. Specifically,   $\beta_{0,0}^{(1)}=2$, as there are two distinct cones:  one closed-cone forming the cross-cycle with vertices $c_0^1(5),c_0^1(6),c_0^1(8),c_0^2(13)$ and one open cone composed of the vertices  $c_0^{1}(4), c_0^{2}(10)$ and $c_0^{1}(5)$. The  vertex $c_0^{2}(10)$ connecting  the two clusters of nodes on layer $1$ and the vertex $c_0^{2}(13)$ are harmonic cross-hubs. 
The Betti number $\beta_{0,0}^{(2)}$ is  equal to $0$, since there are no cones with a vertex on layer $1$ and edges on layer $2$.
 \begin{figure}[t]
\centering
\includegraphics[width=8.5cm,height=3.7cm]{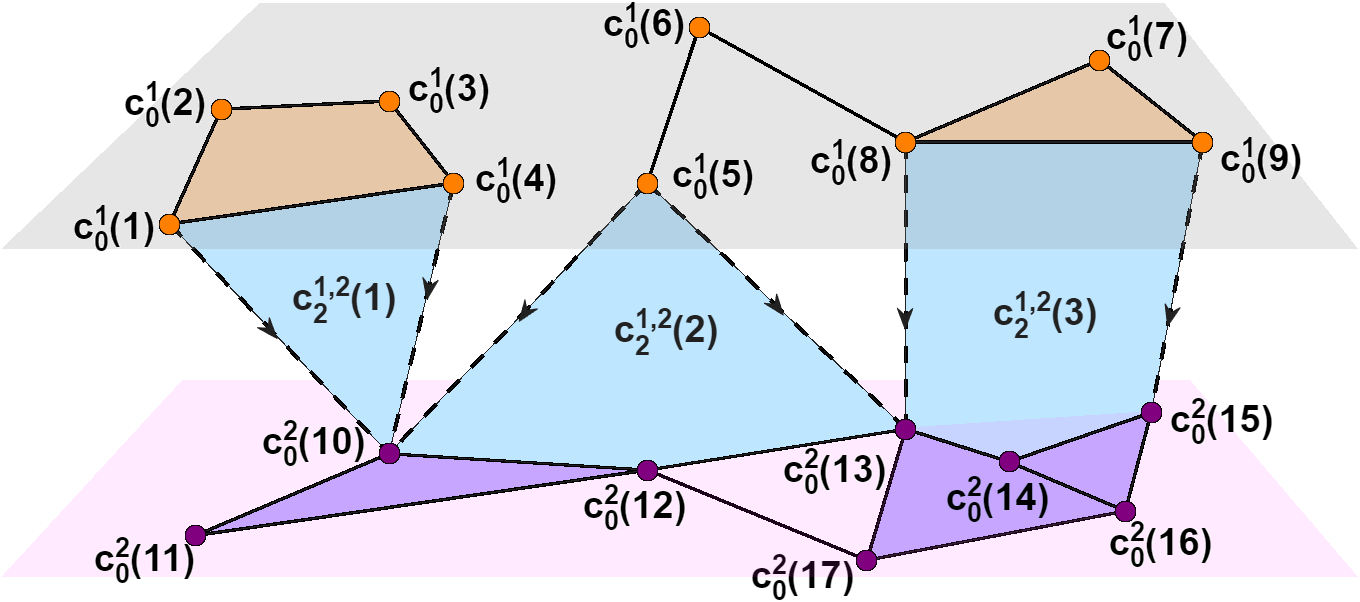}
\caption{Pictorial representation of a $2$-order CMC. }
\label{fig:CMC_2}
\end{figure}

\section{Signal Processing over Cell MultiComplexes}
\label{sec: signal_proc}
As  shown  the cross-Laplacians  are algebraic operators able to capture the topological structure of the CMCs. Then, they are suitable algebraic tools   for  representing signals defined over CMCs.
Let us consider a $2$-order CMC $\mathcal{X}$. When viewed as  a monocomplex structure, we have $\mathcal{X}=(\mathcal{V},\mathcal{E},\mathcal{C})$, where $\mid \mathcal{V}\mid=N$,  $\mid \mathcal{E}\mid=E$ and  $\mid \mathcal{C}\mid=C$  are the dimensions of the node, edges and $2$-cells sets, respectively. We can define signals over the set of nodes, edges and $2$-cells as the maps
\beq 
\bs_0: \mathcal{V}\rightarrow \mathbb{R}^{N},  \, \bs_1: \mathcal{E}\rightarrow \mathbb{R}^{E}, \bs_2: \mathcal{C}\rightarrow \mathbb{R}^{C}. \eeq 
Considering the multilayer structure of the CMC, we assume that it is composed of $L$ interconnected layers and, for simplicity of notation, that only pairs of consecutive layers are inter-connected. Then,  the  node signal $\bs_0$ can be written as 
\beq 
\bs_0=[\bs_0^{1};\bs_0^{2};\ldots;\bs_0^{L}]
\eeq
where $\bs_0^{i}$ is the signal defined on the nodes of the $i$-th layer. 
Similarly, the edge signal is represented as 
\beq \label{eq:overall_edge_signal}
\bs_1=[\bs_1^{1};\bs_1^{1,2};\bs_1^{2};\bs_1^{2,3};\bs_1^{3};\ldots; \bs_1^{L}]
\eeq
where $\bs_1^{i}$ and $\bs_1^{\ell,m}$ are the signals defined on the intra-layer edges and on the cross-edges, respectively.
Finally, we define the signal observed on the $2$-order cells as
\beq
\bs_2=[\bs_{2}^{1};\bs_2^{1,2};\bs_2^{2};\bs_2^{2,3};\bs_2^{3};\ldots; \bs_2^{L}]
\eeq
where $\bs_{2}^{i}$ and $\bs_2^{\ell,m}$ are the signals defined on the intra-layer $2$-order cells and on the cross-cells $c_{2}^{\ell,m}$, respectively.\\
\textbf{Multiple Hodge-based signal decompositions}.
The key novelty and strength of the proposed approach lies in the use of the cross-Laplacians, which enable multiple representations of the same signal based on the local features we aim to learn.
Using  the Hodge decompositions in (\ref{eq: Hodge_space_dec_ell}), (\ref{eq: Hodge_space_dec_m}), 
we can use multiple decompositions of the signal, each partitioning it into different orthogonal components. \\
In this first study, we 
 focus on the  $(0,0)$-cross Laplacians in (\ref{ew_L_00_ell})  and  (\ref{ew_L_00_m}).  It follows directly from (\ref{eq: Hodge_space_dec_ell}),  that the cross-edges signal $\bs_{1}^{\ell,m}$, belonging to the space $\mathbb{R}^{N_{0,0}^{\ell,m}}$, can be decomposed  as
\beq \label{eq:Hodge_L00_ell}
\bs_1^{\ell,m} = \mB_{0,0}^{(\ell),m \, T} \bs_0^{m}+ \mB_{1,0}^{(\ell),m} \bs_2^{\ell,m}+\bs_{1,H}^{\ell,m},
\eeq
where the node signal $\bs_0^{m}\in \mathbb{R}^{N_{-1,1}^{\ell,m}}$ is observed over the nodes within layer $m$ and $\bs_2^{\ell, m}\in \mathbb{R}^{N_{1,0}^{\ell,m}}$ is a $2$-order signal observed over the filled  $(1,0)$ cross-cells $c_{2}^{\ell,m}$ between layers $\ell,m$.  Finally, the harmonic edge signal  $\bs_{1,H}^{\ell,m}$ belongs to the subspace spanned by  $\text{ker}(\mL_{0,0}^{(\ell),m})$, whose dimension is the number of  $(1,0)$ cones between the two layers.  Equivalently, 
 from (\ref{eq: Hodge_space_dec_m}) we get the Hodge decomposition
\beq \label{eq:Hodge_L00_m}
\bs_1^{\ell,m} = \mB_{0,0}^{\ell,(m) \, T} \bs_0^{\ell}+ \mB_{0,1}^{\ell,(m)} \bs_2^{\ell,m}+\bs_{1,H}^{\ell,m},
\eeq
where the node signal $\bs_0^{\ell}\in \mathbb{R}^{N_{1,-1}^{\ell,m}}$ is observed over the nodes within layer $\ell$ and $\bs_2^{\ell, m}\in \mathbb{R}^{N_{0,1}^{\ell,m}}$ is a $2$-order signal observed over the filled  $(0,1)$ cross-cells $c_{2}^{\ell,m}$ between layers $\ell,m$.\\
The physical-based interpretation of the flow divergence and circulation (curl)  offered by the Hodge representation of the   monolayer cell complex structure, can be also generalized to the Hodge signal decompositions in (\ref{eq:Hodge_L00_ell}) and (\ref{eq:Hodge_L00_m}). In this interpretation, it is important to consider that we are adopting a local topological perspective, specifically, one in which a layer is viewed from the standpoint of another connected layer. This implies  the introduction of cross-divergence and   cross-curl operators.
Therefore, for the flows between layers $\ell$ and $m$, using  for example (\ref{eq:Hodge_L00_ell}), we can define 
the cross-divergence \beq \text{div}_{cr}(\bs_{1}^{\ell,m})=\mB_{0,0}^{(\ell),m} \bs_1^{\ell,m} \eeq that is a node signal measuring the conservation of the cross-flows  over the nodes of layer $m$. The cross-curl term \beq \text{curl}_{cr}(\bs_{1}^{\ell,m})= \mB_{1,0}^{(\ell),m\, T} \bs_1^{\ell,m} \eeq is instead a measure of the flow conservation along cross-edges bounding $(1,0)$ cross-cells. 
Therefore, we define the first term   in (\ref{eq:Hodge_L00_ell}), $\bs_{{cr}\text{-}{irr}}^{(\ell),m}:=\mB_{0,0}^{(\ell),m \, T} \bs_0^{m}$, as the cross-irrotational flow.
It has  zero cross-curl along the edges of   $(1,0)$  cross-cells. 
 The second flow in (\ref{eq:Hodge_L00_ell}),   defined as $\bs_{{cr}\text{-}{sol}}^{(\ell),m}:=\mB_{1,0}^{(\ell),m} \bs_2^{\ell,m}$, has zero-sum on the vertices over layer $m$. The flow $\bs_{{cr}\text{-}{sol}}^{(\ell),m}$ is the cross-solenoidal flow.
\vspace{-0.5cm}
\subsection{Spectral cross-signal representation}
\label{sec: spectral_theory}
By viewing the cell complex structure as a monolayer structure, we can leverage the 
spectral theory developed in \cite{sardellitti2024topological}
to represent topological signals using as bases  the eigenvectors of the Hodge Laplacian matrix. Considering a cell complex of order two, the first-order Hodge Laplacian $\mL_1$ admits the eigen-decomposition $\mL_1=\mU_1 \boldsymbol{\Lambda}_1 \mU_1^T$ where the columns of  $\mU_1$ are the eigenvectors  and the diagonal entries of $\boldsymbol{\Lambda}_1$ are the associated eigenvalues.
Therefore, the edge signal $\bs_1$ in (\ref{eq:overall_edge_signal}) can be represented using  the unitary bases $\mU_1 \in \mathbb{R}^{E\times E}$ as
\beq
\bs_1=\mU_1 \hat{\bs}_1
\eeq
where $\hat{\bs}_1$ 
are the Cell Fourier Transform (CFT) coefficients with inverse CFT defined as $ \hat{\bs}_1=\mU_1^T\bs_1$. A cross-signal is called bandlimited with bandwidth $W$ if it admits a sparse representation, i.e. it can be represented using $W$ eigenvectors.

Suppose now that our goal is to process the flow $\bs_{1}^{\ell,m}$ between two layers, since we are interested in capturing the information exchanged between two different systems. The edge signal $\ms_{1}^{\ell,m}$  consists of a subset of components of the  edge signal $\bs_1$ defined in
(\ref{eq:overall_edge_signal}) and, thus lies  in a subspace of $\mathbb{R}^E$. This implies that representing $\bs_{1}^{\ell,m}$ using the orthogonal bases $\mU_1$ spanning the overall space $\mathbb{R}^E$, results in an overcomplete representation of the signal. To overcome this issue, we can use the cross-Laplacian matrix whose eigenvectors induce the Hodge partition of the cross-signal space  in three orthogonal components. \\
To extend the cell complex spectral theory \cite{sardellitti2024topological} to CMCs, we focus w.l.o.g. on the cross-Laplacian $\mL_{0,0}^{(\ell),m}$. Given the eigendecomposition $\mL_{0,0}^{(\ell),m}=\mU_{0,0}^{(\ell),m} \boldsymbol{\Lambda}_{0,0}^{(\ell),m}\mU_{0,0}^{(\ell),m\, T}$, we can represent cross-edge signals on the subspace spanned by the eigenvectors of the cross-Laplacian. Hence, we define the CMC Fourier Transform as the projection of a cross-edge signal $\bs_{1}^{\ell,m}$ onto the space spanned by the eigenvectors of $\mL_{0,0}^{(\ell),m}$, i.e.   $\hat{\bs}_{1}^{\ell,m}:= \mU_{0,0}^{(\ell),m\, T}\bs_{1}^{\ell,m}$. Then, the cross-edge signal can be represented as a vector belonging to  $\mathbb{R}^{N_{0,0}^{\ell,m}}$  with $\mathbb{R}^{N_{0,0}^{\ell,m}} \subset \mathbb{R}^{E}$ as
\beq
\bs_{1}^{\ell,m}:= \mU_{0,0}^{(\ell),m}\hat{\bs}_{1}^{\ell,m}.
\eeq
It is important to remark that according to the goal of our learning task,  different signal dictionaries can be used to extract different signal features. 
\begin{figure*}[t]
    \centering
    \begin{subfigure}[b]{0.28\textwidth}
        \includegraphics[width=1.0\textwidth, height=3.2cm]{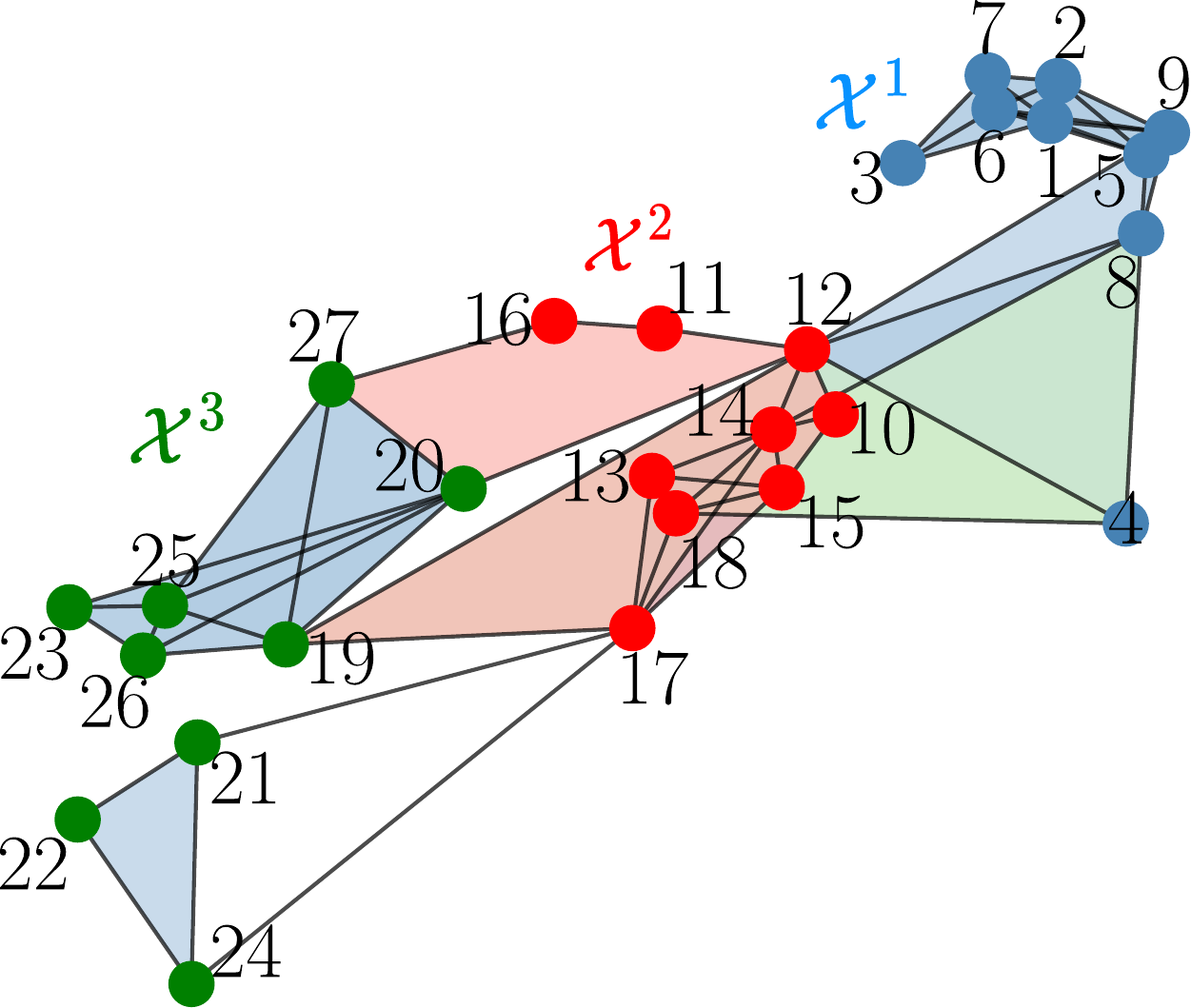}
        \caption{CMC with two holes}
    \end{subfigure}
    \hfill
    \begin{subfigure}[b]{0.28\textwidth}
        \includegraphics[width=1.0\textwidth, height=3.2cm]{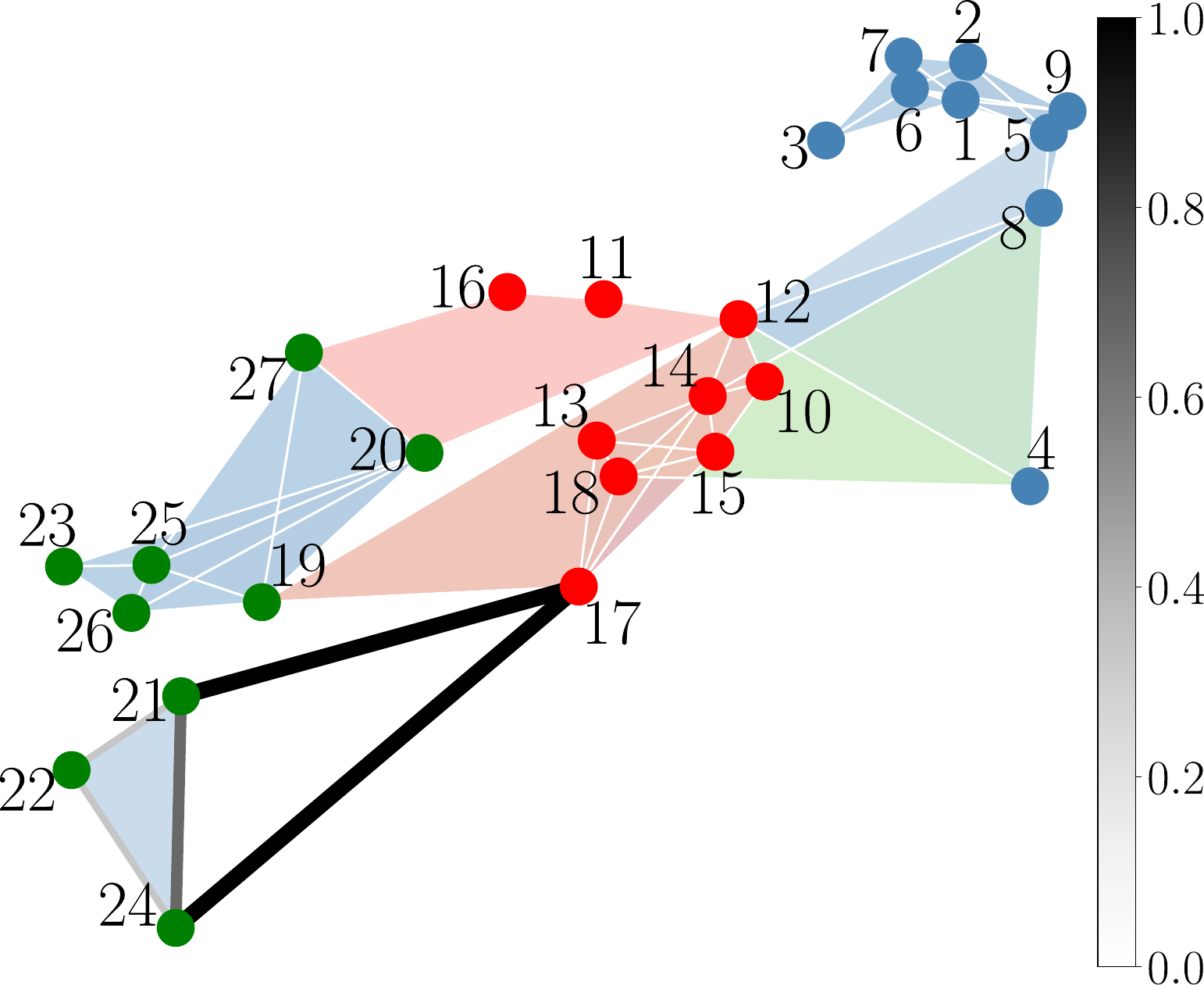}
        \caption{First harmonic eigenvector of $\mL_1$}
    \end{subfigure}
    \hfill
    \begin{subfigure}[b]{0.28\textwidth}
        \includegraphics[width=1.0\textwidth, height=3.2cm]{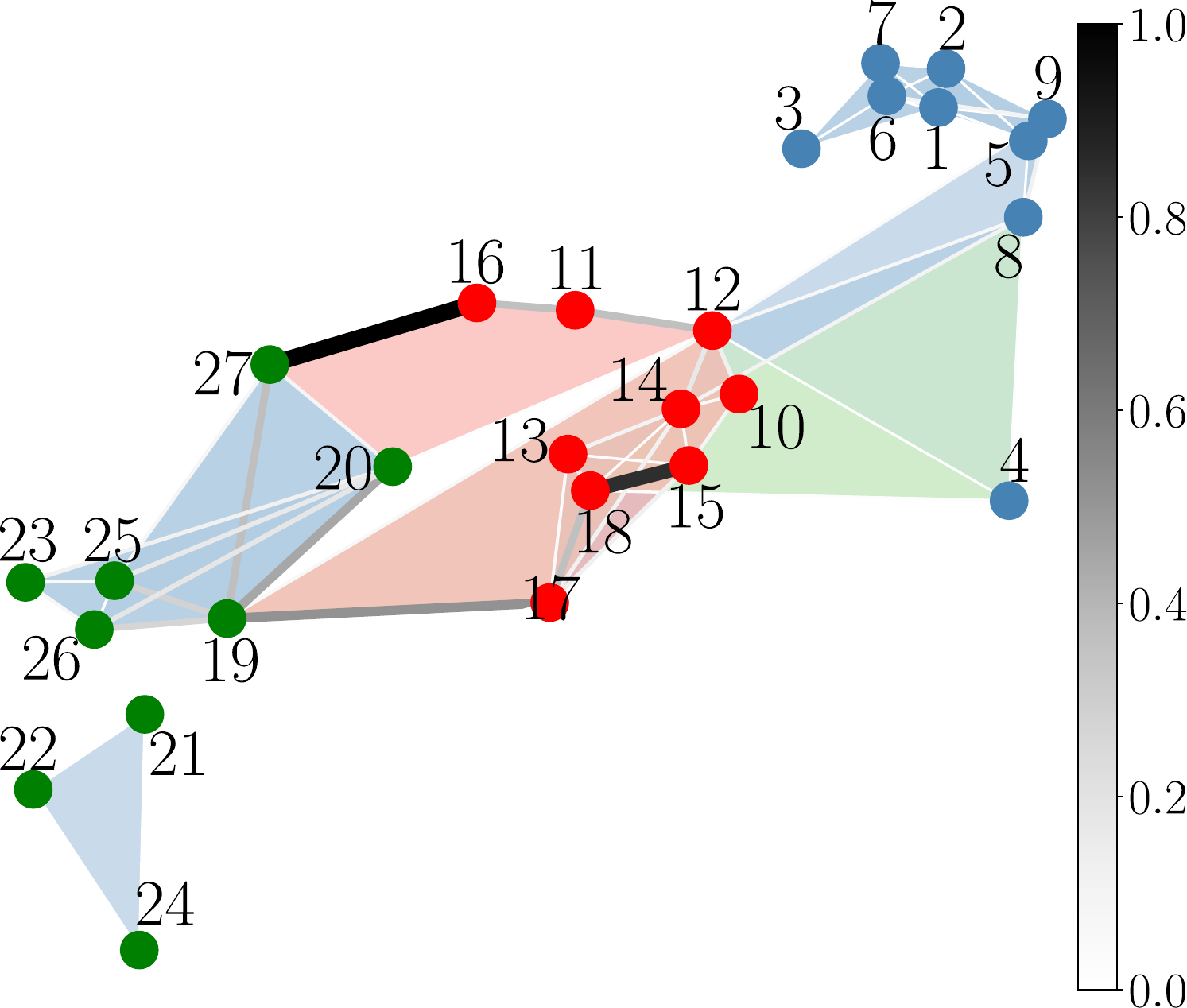}
          \caption{Second harmonic eigenvector of $\mL_1$}
    \end{subfigure}

    \vspace{0.2cm}

    \begin{subfigure}[b]{0.28\textwidth}
        \includegraphics[width=1.0\textwidth, height=3.0cm]{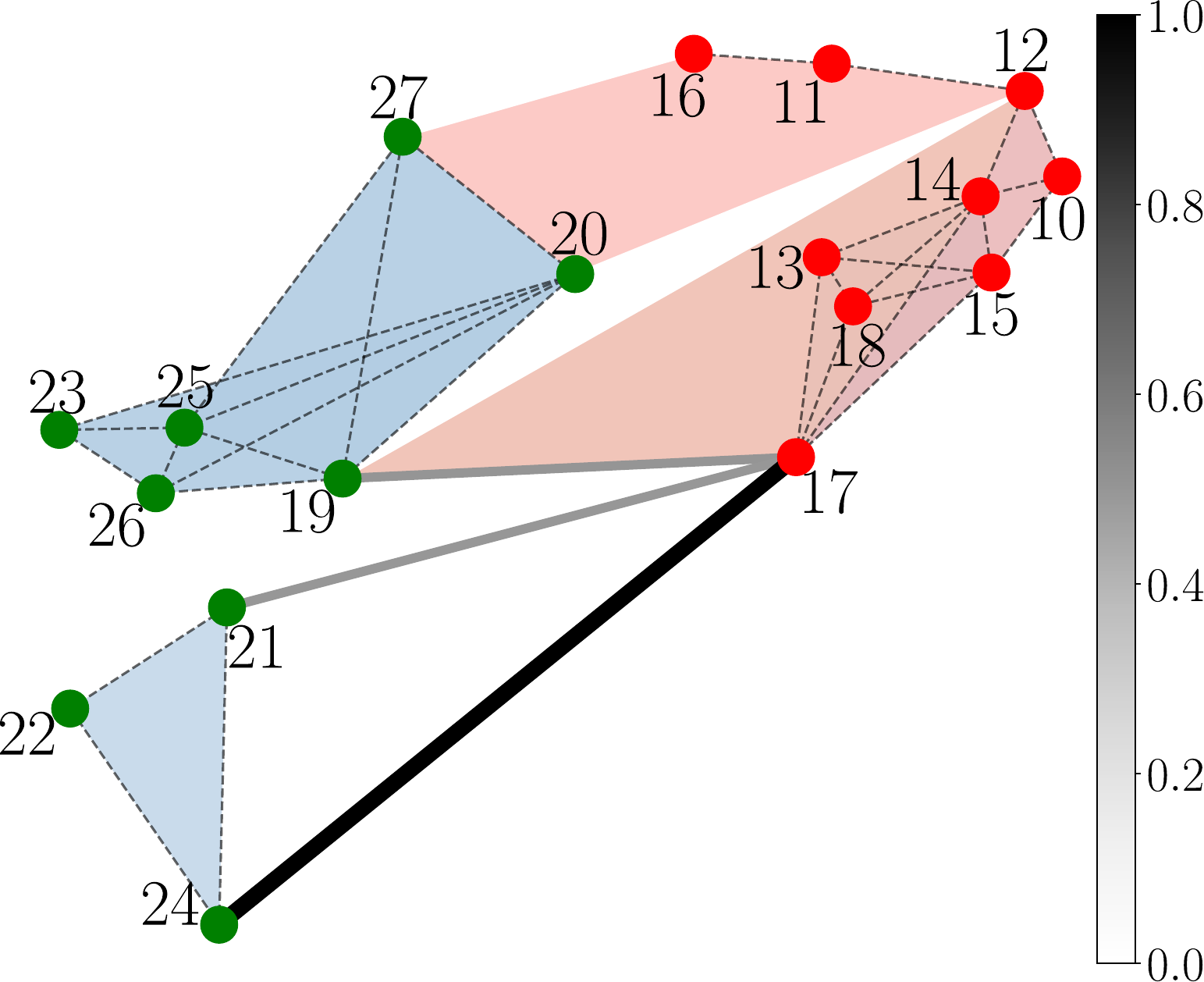}
          \caption{First harmonic eigenvector of $\mL_{0,0}^{2,(3)}$}
    \end{subfigure}
    \hfill
    \begin{subfigure}[b]{0.3\textwidth}
        \includegraphics[width=1.0\textwidth, height=3.0cm]{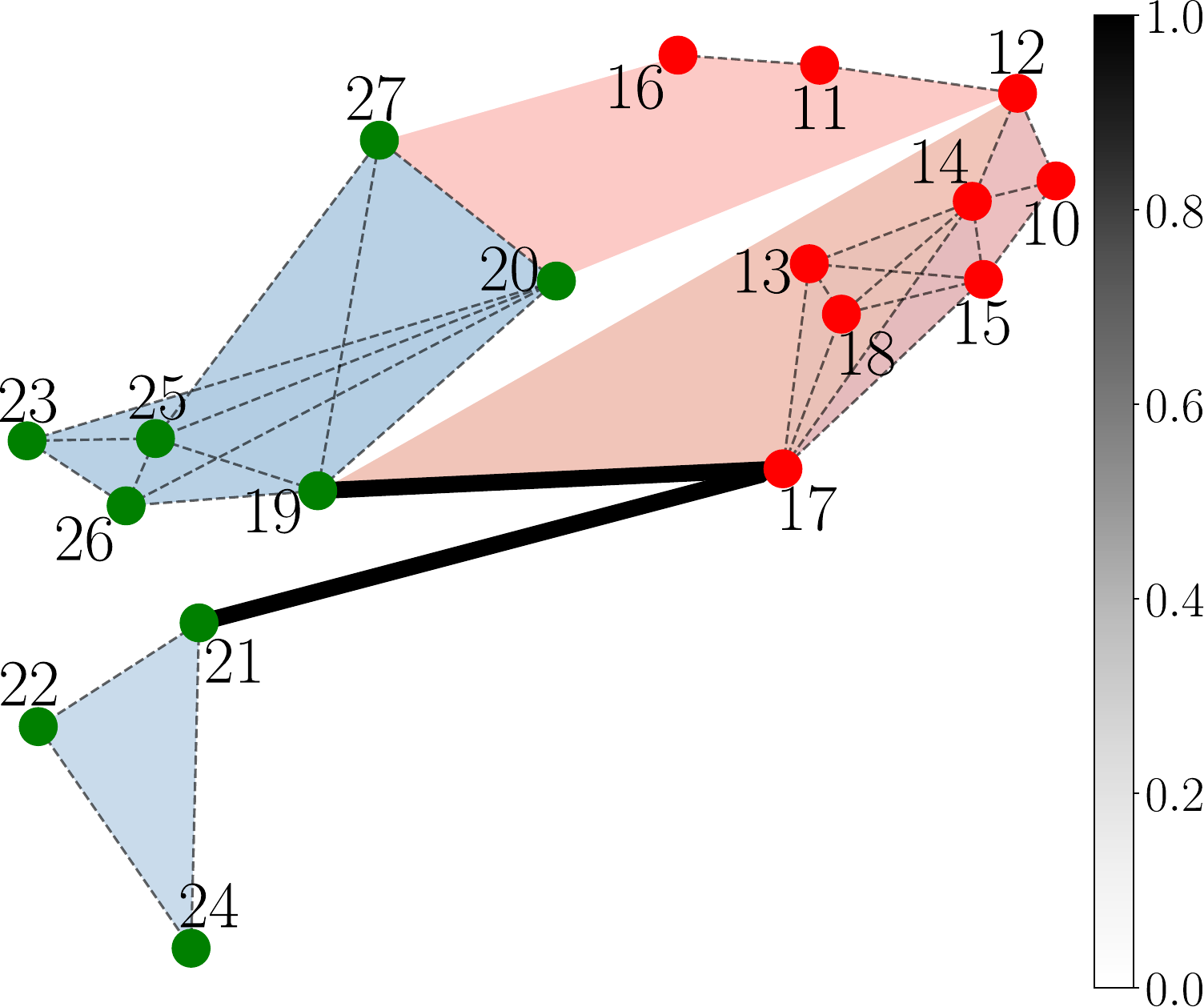}
          \caption{Second harmonic eigenvector of $\mL_{0,0}^{2,(3)}$}
    \end{subfigure}
    \hfill
    \begin{subfigure}[b]{0.3\textwidth}
        \includegraphics[width=1.0\textwidth, height=3.0cm]{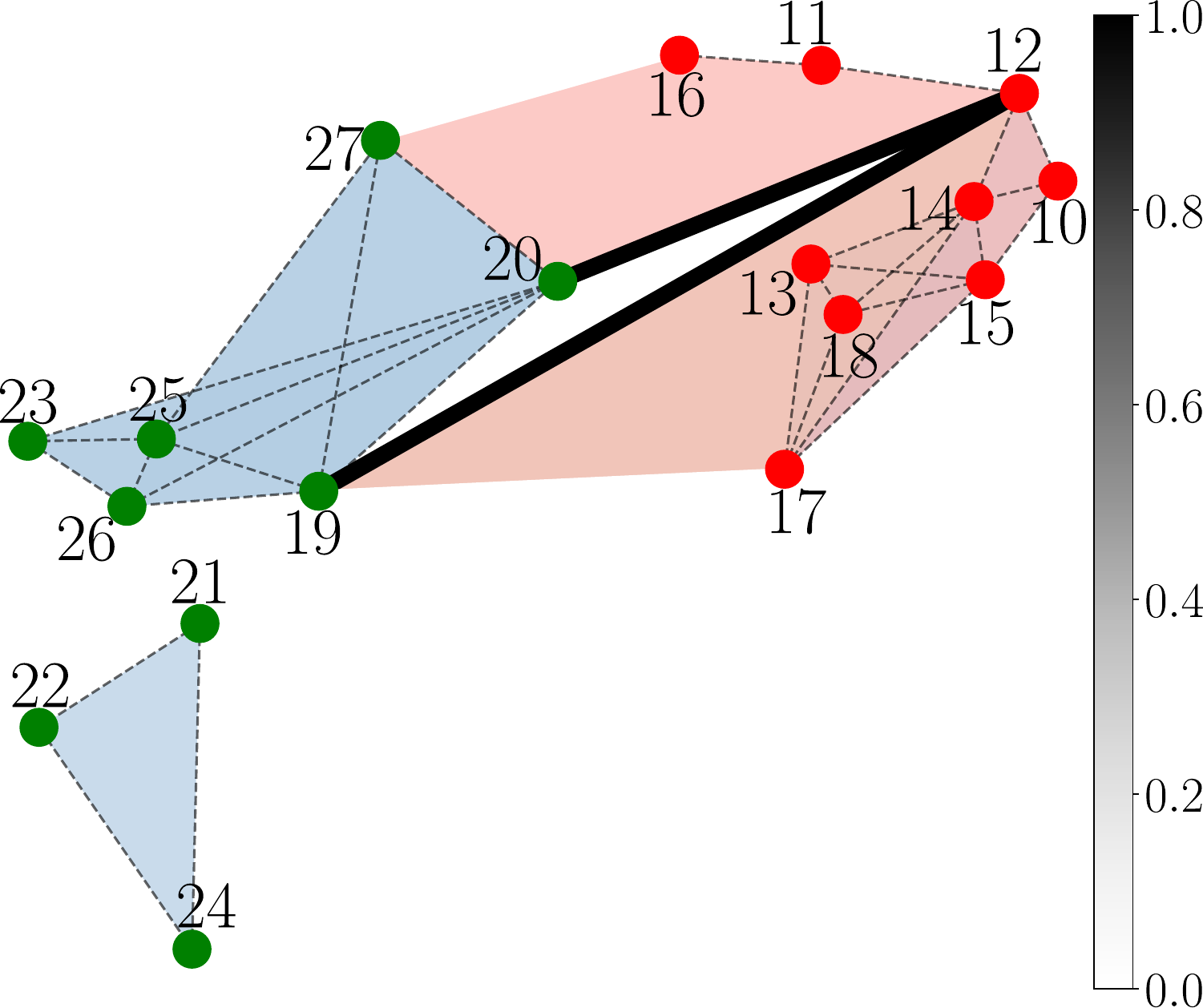}
        \caption{Third harmonic eigenvector of $\mL_{0,0}^{2,(3)}$}
    \end{subfigure}

    \caption{(a) Illustrative example of a $3$-layers CMC. Magnitude of the harmonic eigenvectors associated with: (b) and (c) the monolayer Hodge Laplacian $\mL_{1}$; (d), (e) and (f) the cross-Laplacian $\mL_{0,0}^{2,(3)}$.}
    \label{fig:harmonic_eigenvectors}
\end{figure*}\\
\textbf{Numerical results.}
To give more insights on how the cross-Laplacian bases reflect the invariants of the space, we illustrate in the example in Fig.  \ref{fig:harmonic_eigenvectors} the harmonic eigenvectors of both the monolayer first-order Laplacian matrix $\mL_1$  and the cross-Laplacian matrix $\mL_{0,0}^{2,(3)}$ considering the  CMC illustrated in  Fig.  \ref{fig:harmonic_eigenvectors}(a). The complex has two triangular holes $h_1,h_2$ identified by the nodes $h_1=(19,20,12)$ and $h_2=(24,21,17)$, then the harmonic subspace of $\mL_1$ has dimension $2$. 
In Figs. \ref{fig:harmonic_eigenvectors}(b) and (c) we represent over the edges the absolute magnitude of the first- and second-harmonic eigenvectors of $\mL_{1}$. It can be observed that  both eigenvectors tend to identify cycles around holes, albeit the first eigenvector is more tightly localized around the hole $h_2$. In Figs. \ref{fig:harmonic_eigenvectors}(d), (e) and (f)  we represent the first, second and third harmonic eigenvectors of $\mathbf{L}_{0,0}^{2,(3)}$, respectively. Note that the harmonic subspace of $\mL_{0,0}^{2,(3)}$ is spanned by three eigenvectors since in the cross-complex between layers $2$ and $3$ we have  the two closed cones $h_1$ and $h_2$ and the open cone $h_3=(24,21,17)$. It can be observed as all the eigenvectors are highly localized around the holes and the cones of the cross-complex $\mathcal{X}^{2,3}$.

Finally, to investigate the main benefits of the local processing in terms of signal space dimensionality reduction, we consider the problem of finding a sparse representation for a set of observed cross-edge signal $\by_1^{\ell,m}(i)$, $i=1,\ldots,M$, with a prescribed signal representation error. Then, we consider two cases: i) we assume a global monocomplex  representation of the CMC by finding a sparse representation of the signal $\by_1^{\ell,m}$  over the eigenvector bases $\mU_1$; ii) we consider a CMC multilayer structure to find a  sparse representation of the cross-edge signal $\by_1^{\ell,m}$  over the eigenvector bases $\mU_{0,0}^{(\ell),m}$ (or $\mU_{0,0}^{\ell,(m)}$). In both cases, we solve a basis pursuit problem with the goal of finding an optimal trade-off between the sparsity of the representation and the data-fitting error. Hence, in case i) for any given observed edge vector $\by_1(m)\in \mathbb{R}^E$, $m=1,\ldots,M$, we derive the sparse vector $\hat{\bs}_1$ as solution of the following basis pursuit problem
\cite{Donoho98}:
\beq \label{eq:bas_pur_L1}
\begin{array}{lll}
 \underset{\hat{\bs}_1 \in \mathbb{R}^E}{\text{min}} & \parallel
\hat{\bs}_1\parallel_1   \qquad \qquad (\mathcal{B}_1)\\
 \; \; \text{s.t.} & \parallel
 {\by}_1 -\mU_1 \hat{\bs}_1\parallel_F \leq \epsilon
 \end{array}
\eeq
 where $\epsilon$ is a required bound on the signal fitting error. Similarly, in case ii) we solve the following problem 
\beq \label{eq:bas_pur_L00}
\begin{array}{lll}
 \underset{\hat{\bs}_1^{\ell,m} \in \mathbb{R}^{N_{0,0}^{\ell,m}}}{\text{min}} & \parallel
\hat{\bs}_1^{\ell,m}\parallel_1   \qquad \qquad (\mathcal{B}_{0,0})\\
 \; \; \text{s.t.} & \parallel
 {\by}_1^{\ell,m} -\mU_{0,0}\hat{\bs}_1^{\ell,m}\parallel_F \leq \epsilon_{0,0}
 \end{array}
\eeq
where $\mU_{0,0}=\mU_{0,0}^{(\ell),m}$ (or $\mU_{0,0}^{\ell,(m)}$) and $\epsilon_{0,0}=\epsilon \, N_{0,0}^{\ell,m}/E$ is a maximum value on the fitting error scaled according to the signal dimension. 

As numerical test, we derive  the normalized mean squared error defined as $\text{NMSE}=\frac{\parallel {\by}_1^{\ell,m}-\overline{{\by}}_1^{\ell,m}\parallel_F }{\parallel {\by}_1^{\ell,m}\parallel_F}$ 
 using the estimated cross-signals $\bar{{\by}}_1^{\ell,m}$ obtained by solving both problems $\mathcal{B}_1$ and  $\mathcal{B}_{0,0}$. We generate the observed edge signals $\by_1(i)$, $i=1,\ldots,M$, from a standard Gaussian distribution. Specifically,  we focus on the cross-edge signal between layers $\ell=2$ and $m=3$ of the CMC illustrated in Fig.  \ref{fig:harmonic_eigenvectors}(a) and then consider the cross-Laplacian matrix $\mL_{0,0}^{2,(3)}$. We estimate   first the global edge signal $\bar{{\by}}_1=\mU_1 \hat{\bs}_1$ by solving problem $\mathcal{B}_1$ for each observed signal $\by_1(i)$. Then, we  extract the estimated  cross-edge signal as $\bar{{\by}}_1^{2,3}=\mathbf{D}^{2,3}\bar{{\by}}_1$  where $\mathbf{D}^{2,3}$ is a diagonal matrix with entries equal to $1$  for the cross-edges between layer  $2,3$ and $0$ otherwise. Hence, we solve the problem $\mathcal{B}_{0,0}$ considering the observed cross-edge signals ${{\by}}_1^{2,3}=\mathbf{D}^{2,3}{{\by}}_1$  to derive the  optimal solution $\bar{{\by}}_1^{2,3}=\mU_{0,0}^{2,(3)}\hat{\bs}_1^{2,3}$.
 In Fig. \ref{fig:sparsity} we illustrate the NMSE averaged over $M=1000$ edge signal realizations versus the sparsity of the signal. The sparsity is measured as the number of eigenvectors of $\mU_1$ and $\mU_{0,0}^{2,(3)}$ that are used as dictionary bases by solving, respectively, problems $\mathcal{B}_{1}$ and $\mathcal{B}_{0,0}$.  It can be noticed that the cross-Laplacian eigenvectors provide an efficient dictionary basis to represent the local signal $\by_1^{2,3}$,  achieving a better accuracy/sparsity  trade-off with respect to the overcomplete dictionary provided by the eigenvectors of the first-order Laplacian  of the global structure.   
\begin{figure}[t]
\centering
\includegraphics[width=8.5cm,height=5.0cm]{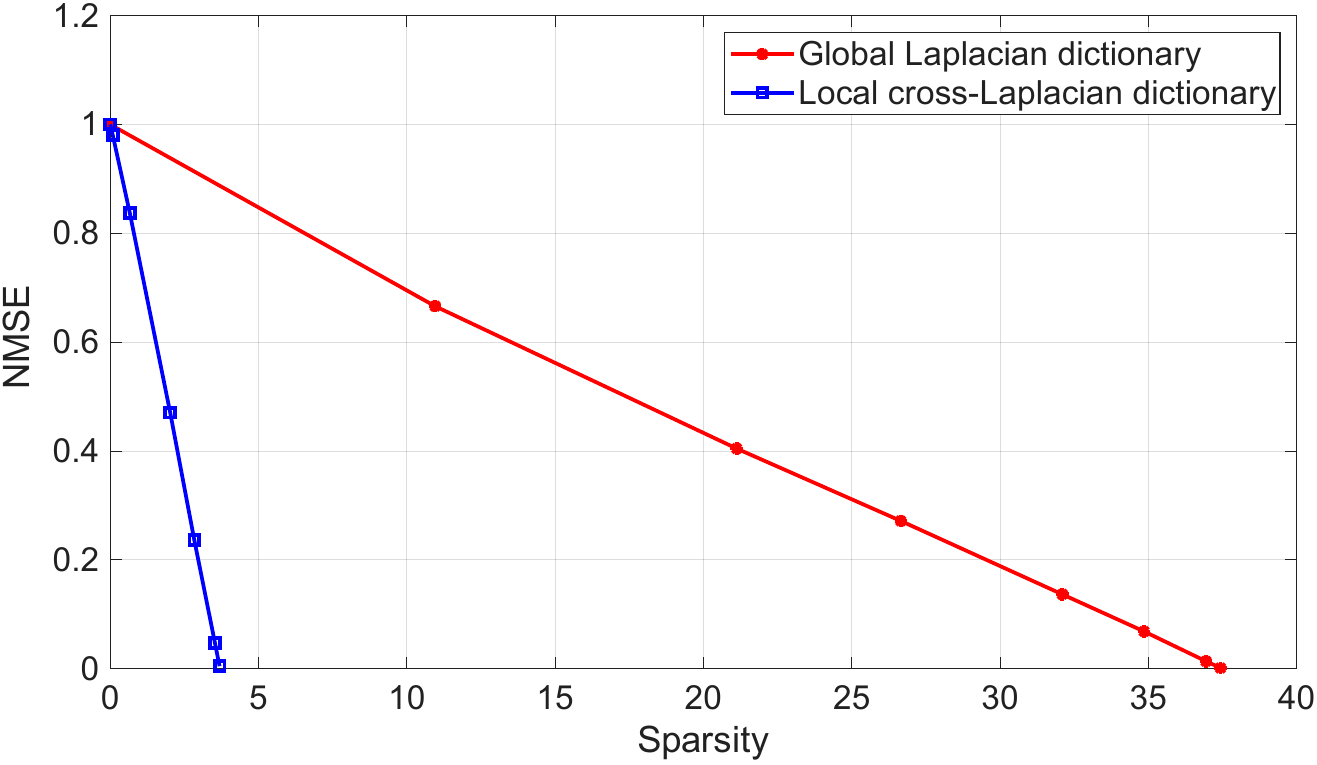}
\caption{NMSE versus signal sparsity.}
\label{fig:sparsity}
\end{figure}
\vspace{-0.3cm}

\section{Estimation of cross-edge signals}
\label{sec: signal_estim}
In this section we focus on the optimal estimation of the cross-edge signals represented by the Hodge decompositions in (\ref{eq:Hodge_L00_ell}) and (\ref{eq:Hodge_L00_m}).   We design optimal signal estimators from observed noisy cross-signals expressed as
$\by_1^{\ell,m}=\bs_1^{\ell,m}+\bn_1$
where the additive noise vector $\bn_1$  follows a Gaussian distribution, i.e. $\bn_1 \sim \mathcal{N}(\mathbf{0}, \sigma_n^2 \mathbf{I})$. The optimal node, $2$-cells and harmonic   signals, can be derived \cite{barb_2020} as the solutions of the following problem
\beq \nonumber
\begin{array}{lllll}
\!\!\underset{\vspace{0.04cm}\!\substack{ \vspace{0.07cm}\\\bs_{0}^{\ell}\in {\mathbb{R}}^{N_{\ell}^{}},\bs_{2}^{\ell,m} \in \mathbb{R}^{ {N_{0,1}^{\ell,m}}}  \\ \bs_{1,H}^{\ell,m}  \scriptstyle{\in \mathbb{R}^{N_{0,0}^{\ell,m}}  }}}{\min \medskip}  \!\!\!\!\!\! \!\!\!\!\!\! \parallel  \mB_{0,0}^{\ell,(m) \, T} \!\bs_0^{\ell}\!\!+ \mB_{0,1}^{\ell,(m)} \bs_2^{\ell,m}\!\!+\bs_{1,H}^{\ell,m} -\by_1^{\ell,m} \parallel^2 \medskip\\
\quad \quad \text{s.t.} \quad  \quad \mB_{0,0}^{\ell,(m)} \bs_{1,H}^{\ell,m}=\mathbf{0}, \quad \mB_{0,1}^{\ell,(m)  \, T} \bs_{1,H}^{\ell,m}=\mathbf{0}   \quad \quad (\mathcal{P}_u).
\end{array}
\eeq
It can be easily proved that this problem admits the following closed-form solutions \cite{barb_2020}:
\beq
\label{eq:closed_form}
\begin{array}{lll}
\bar{\bs}_{0}^{\ell}=(\mB_{0,0}^{\ell,(m) } \mB_{0,0}^{\ell,(m)\, T})^{\dagger} \mB_{0,0}^{\ell,(m)} \by_{1}^{\ell,m}\\
\bar{\bs}_{2}^{\ell,m}=(\mB_{0,1}^{\ell,(m) \, T} \mB_{0,1}^{\ell,(m)})^{\dagger} \mB_{0,1}^{\ell,(m)\, T} \by_{1}^{\ell,m}\\
\bar{\bs}_{1,H}^{\ell,m}=\by_1^{\ell,m}-
\mB_{0,0}^{\ell,(m) \, T}\bar{\bs}_{0}^{\ell}-\mB_{0,1}^{\ell,(m)}\bar{\bs}_{2}^{\ell,m}
\end{array}
\eeq
where ${}^{\dagger}$ denotes the Moore-Penrose pseudo-inverse. 

As  numerical test, we consider the CMC illustrated in Fig.  \ref{fig:harmonic_eigenvectors}(a)  composed of $N=27$, $E=61$ and $C=49$ nodes, edges and $2$-order cells, respectively. We filled the $(0,1)$ cross-cell with vertices $21,24,17$.  Our goal is to estimate the cross-edge signal between layer $2$ (red nodes) and layer $3$ (green nodes). We consider the $\mL_{0,0}^{2,(3)}$ cross-Laplacian  and the associated Hodge decomposition of the cross-edge signal given in (\ref{eq:Hodge_L00_m}). Hence,  we generate $M=5000$ random Gaussian edge signal vectors $\by_1(i)\sim \mathcal{N}(\mathbf{0},\sigma_1^2 \mathbf{I})$ for $i=1,\ldots,M$, assuming the model in (\ref{eq:Hodge_L00_m}) for the cross-edge signal. Therefore, we estimate the cross-solenoidal and cross-irrotational components of the  signals using the closed form solutions in (\ref{eq:closed_form}), obtaining $\bar{\by}_{cr\text{-}irr}^{2,3}(i)=\mB_{0,0}^{2,(3)} \bar{\bs}_0^{2}(i)$ and 
$\bar{\by}_{cr\text{-}sol}^{2,3}(i)={\mB_{0, 1}^{2,(3)}}^T \bar{\bs}_2^{2,3}(i)$.
We report in Fig. \ref{fig:CMC_4}  the average normalized squared error $\text{NMSE}:= \sum_{i=1}^{M}\frac{\parallel \bar{\by}_1^{2,3}(i)-{\by}_1^{2,3}(i)\parallel}{ \parallel {\by}_1^{2,3}(i) \parallel M}$ versus the signal-to-noise ratio $\text{SNR}=\sigma_1^{2}/\sigma_n^{2}$ for both the cross-irrotational and cross-solenoidal signals. 
We compare our method with a monocomplex representation of the topological domain. Specifically, we consider the Hodge-Laplacian matrix $\mL_1$ associated with the overall cell complex and solve the associated problem $\mathcal{P}_u$ to find the estimated irrotational signal  $\bar{\by}_{1,irr}(i)=\mB_1^T \bar{\bs}_{0}(i)$  and the solenoidal component $\bar{\by}_{1,sol}(i)=\mB_2 \bar{\bs}_{2}(i)$. Therefore, we derive the  average  NMSE in the estimation of the cross-edge signal between layers $2$ and $3$, for both the irrotational and solenoidal components.
From Fig. \ref{fig:CMC_4} we can observe as the cross-Laplacian  based  representation ensures better performance in terms of recovering  error for both the  cross-irrotational and cross-solenoidal signals.


\begin{figure}[t]
\centering
\includegraphics[width=8.3cm,height=5.5cm]{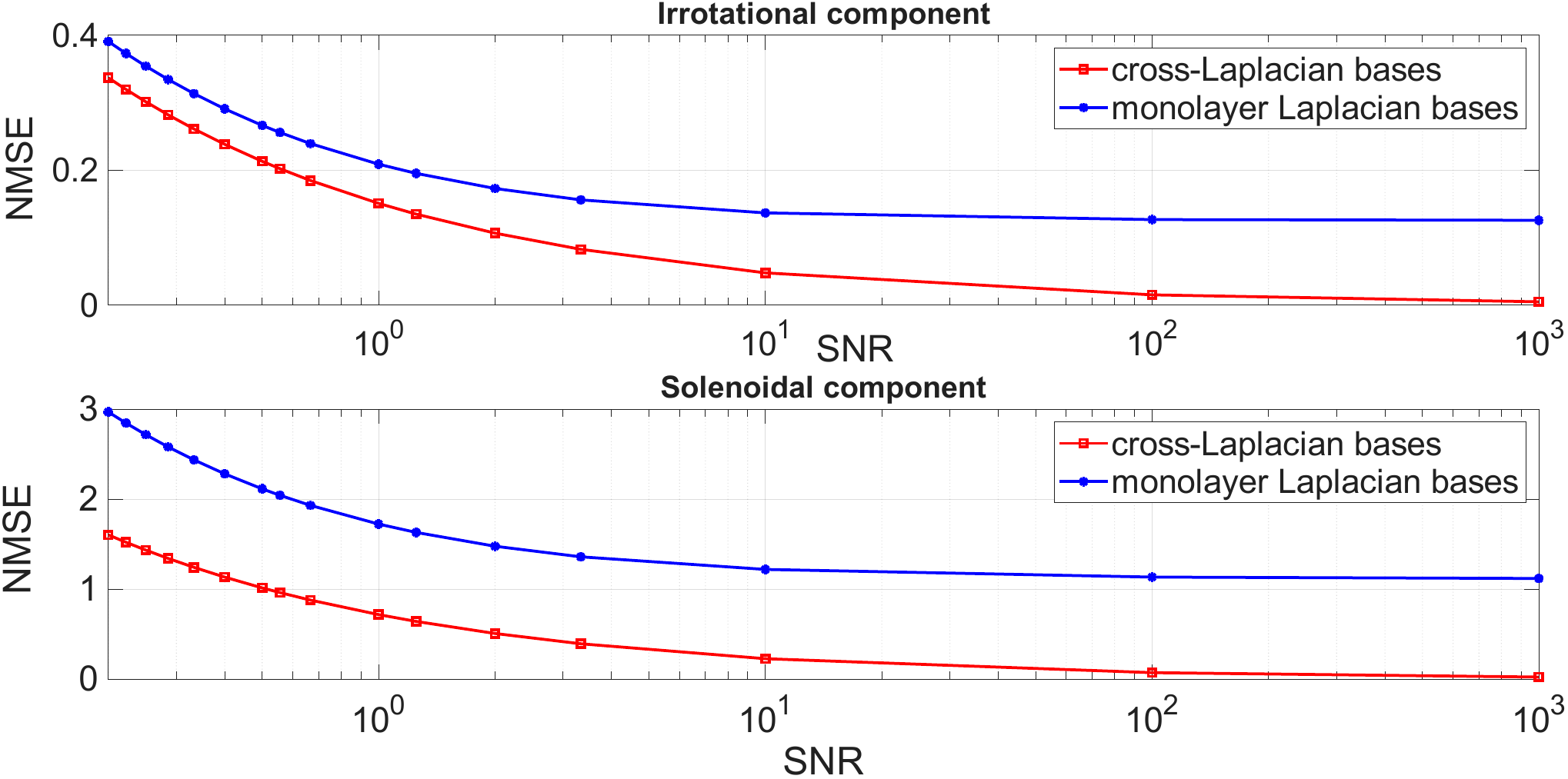}
\caption{Normalized mean squared error versus SNR.}
\label{fig:CMC_4}
\end{figure}

\section{Learning the CMC topology from local data}
\label{sec: signal_learning}
The inference of the topological domain that underlies the observed data is  a key step  for the processing of signals when the topology is unknown. Several works have addressed the problem of learning the simplicial/cell complex structure hinging on Hodge-Laplacian representation of signals \cite{barb_2020},\cite{sardellitti2024topological},\cite{hoppe2024representing},\cite{gurugubelli2024simplicial}. 
By extending the method developed in \cite{sardellitti2024topological} for cell complexes, we propose a strategy to  learn the cell multi-complexes topology based on cross-edge signals smoothness.\\
Let us focus, without loss of generality, on the inference of the $(0,0)$-cross Laplacian matrix in (\ref{ew_L_00_ell})
\beq 
\mL_{0,0}^{(\ell),m}= (\mB_{0,0}^{(\ell),m})^{T} \mB_{0,0}^{(\ell),m}  +\mB_{1,0}^{(\ell),m}(\mB_{1,0}^{(\ell),m})^{T}. \eeq
We assume that the graph between layers $\ell,m$ is known so that our task is learning the $2$-order $(1,0)$ cross-cells between the two layers from the observation of a set of $M$   cross-edge signals $\bx_{1}^{\ell,m}(i)$ for $i=1,\ldots,M$.
This implies that we know the lower cross-Laplacian matrix, i.e. the  matrix $\mB_{0,0}^{(\ell),m}$ and we aim to estimate the upper cross-incidence matrix $\mB_{1,0}^{(\ell),m}$ spanning the cross-curl subspace.

Our first step is to check if the inference of this matrix is meaningful for the observed data, i.e. if the observed signals have components onto the cross-curl and harmonic subspaces. Considering the lower cross-Laplacian matrix  $\mathbf{{L}}_{d \, 0,0}^{(\ell),m}:=(\mB_{0,0}^{(\ell),m})^{T}\mB_{0,0}^{(\ell),m}$, 
we first project the signal  $\mX_{1}^{\ell,m}=[\bx_{1}^{\ell,m}(1),\bx_{1}^{\ell,m}(2),\ldots,\bx_{1}^{\ell,m}(M)]$ onto the subspace orthogonal to the cross-irrotational space, by computing the signal energy, i.e. 
$
\parallel \mX_{1,0}^{\ell,m} \parallel_{F}^{2}=\parallel (\mI-\mU_{d,0}^{(\ell),m}(\mU_{d,0}^{(\ell),m})^T) \mX_{1}^{\ell,m} \parallel_{F}^{2}
$
where $\mU_{d,0}^{(\ell),m}$ are the eigenvectors of $\mathbf{{L}}_{d \, 0,0}^{(\ell),m}$ associated with the non-zero eigenvalues.
Therefore, we compute the ratio $\eta=\parallel \mX_{1,0}^{\ell,m} \parallel_{F}^{2}/\parallel \mX_{1}^{\ell,m} \parallel_{F}^{2}$ and if $\eta$ is lower than   a given threshold, we set $\mB_{1,0}^{(\ell),m}=\mathbf{0}$, otherwise we proceed to learn the matrix $\mB_{1,0}^{(\ell),m}$. Let us now write the matrix $\mathbf{{L}}_{u \, 0,0}^{(\ell),m}=\mB_{1,0}^{(\ell),m} (\mB_{1,0}^{(\ell),m})^T$ 
in the form
\beq \label{eq:L_u_ck}
\mathbf{{L}}_{u \, 0,0}^{(\ell),m}=\sum_{k=1}^{N_{1,0}^{\ell, m}} a_k \mb_{1,0}^{(\ell),m}(k) (\mb_{1,0}^{(\ell),m}(k))^T
\eeq
 where $\mb_{1,0}^{(\ell),m}(k)$ denotes the $k$-th column of the matrix $\mB_{1,0}^{(\ell),m}$ and $a_k$ is a  binary variable assuming the value $1$ if the $(1,0)$ cross-cell of order $2$ is filled, and $0$ otherwise. Finally,
$N_{1,0}^{\ell, m}$ is the number of $2$-order $(1,0)$ cross-cells between layer $\ell,m$.
Hence, our goal is to infer the $(1,0)$ filled cross-cells such that  the variation of the cross-edge signal is minimized, i.e. the signal is smooth onto the upper cross-Laplacian eigenvectors. The cross-variation of the signals can be defined as 
\beq  \label{eq:P_u}
\begin{split}
 \text{CV}(\mX_{1,0}^{\ell,m})&:= \text{tr}\{(\mX_{1,0}^{\ell,m})^T \mathbf{{L}}_{u \, 0,0}^{(\ell),m}\mX_{1,0}^{\ell,m}\}
\end{split}
\eeq
and, using (\ref{eq:L_u_ck}), we  get the following optimization problem 
\beq \label{eq:bas_pur_L00}
\begin{array}{lll}
 \underset{{\ba} \in \{0,1\}^{N_{1,0}^{\ell,m}}}{\text{min}} & \!\!\ds \sum_{k=1}^{N_{1,0}^{\ell, m}} a_k \text{tr}\{(\mX_{1,0}^{\ell,m})^T\mb_{1,0}^{(\ell),m}(k) (\mb_{1,0}^{(\ell),m}(k))^T \mX_{1,0}^{\ell,m}\}   \medskip \\
 \; \; \;\text{s.t.} & \parallel
 \ba \parallel_0 \leq \gamma \qquad \qquad \qquad\qquad (\mathcal{P}_{c})
 \end{array}
\eeq
where $\gamma$ is a coefficient controlling the sparsity of $\ba$, i.e. the number of filled $(1,0)$ cones. Note that albeit problem $\mathcal{P}_{c}$ is non-convex, it admits a closed-form optimal solution. Specifically, defining  $\alpha_k=\text{tr}\{(\mX_{1,0}^{\ell,m})^T\mb_{1,0}^{(\ell),m}(k) (\mb_{1,0}^{(\ell),m}(k))^T \mX_{1,0}^{\ell,m}\}$  and sorting these coefficients in increasing order as $\alpha_{i_1},\alpha_{i_2},\ldots, \alpha_{i_{N_{0,0}^{\ell,m}}}$, the optimal  solution of problem $\mathcal{P}_c$ is obtained by  selecting the $(1,0)$ $2$-order cells  whose cross-circulation  is lower than a given threshold $\gamma$.  Therefore, denoting with  $\mathcal{C}_{s}$ the set of the optimal selected indexes, the  estimated upper cross-Laplacian matrix can be
derived from (\ref{eq:L_u_ck})
setting $a_k=1$ for $k \in \mathcal{C}_{s}$.
To test the effectiveness of the proposed topology inference method, in the following we consider a real-data application by learning higher-order interactions in brain networks.
\vspace{-0.3cm}
\subsection{Application: Inter-modules connectivity in brain networks}

The functional connectivity (FC)  of brain networks is typically organized into modular structures composed of groups of brain regions-of-interest (ROIs)  with highly correlated activity and forming distinct modules with specific functional connectivity patterns \cite{bullmore2009complex}. For the context of our illustrative application, in \cite{arbabyazd2020dynamic}, the authors  introduce meta-connectivity analysis  to identify modules of functional links that  co-vary over time. They also introduce the concept of trimers, i.e. pairs of (inter-modules) links incident on a common root region (meta-hub)   controlling the inter-modules relations. Interestingly, these trimers resemble the $(0,1)$ (or $(1,0)$)  cross-cells of order $2$ between two layers with the meta-hubs functioning as cross-hubs.  Motivated by this  interpretation, it is natural to apply the proposed CMC learning strategy to investigate the brain's second order  structure using real datasets.
We use the Human Connectome Project (HCP) public dataset \cite{HCP}, by selecting $300$ sets of $N=116$ resting-state functional MRI (rs-fMRI) time series obtained from a cohort of $300$ unrelated healthy individuals (young adults aged $21-35$ consisting of $144$ males and $156$ females).
Hence, we derive the edge signals as the correlation coefficient between the time series observed over the vertices of each edge. 
\begin{figure}[t]
    \centering
   
    \begin{subfigure}[b]{\columnwidth}
    \centering
        \includegraphics[width=0.9\textwidth, height=4.8cm]{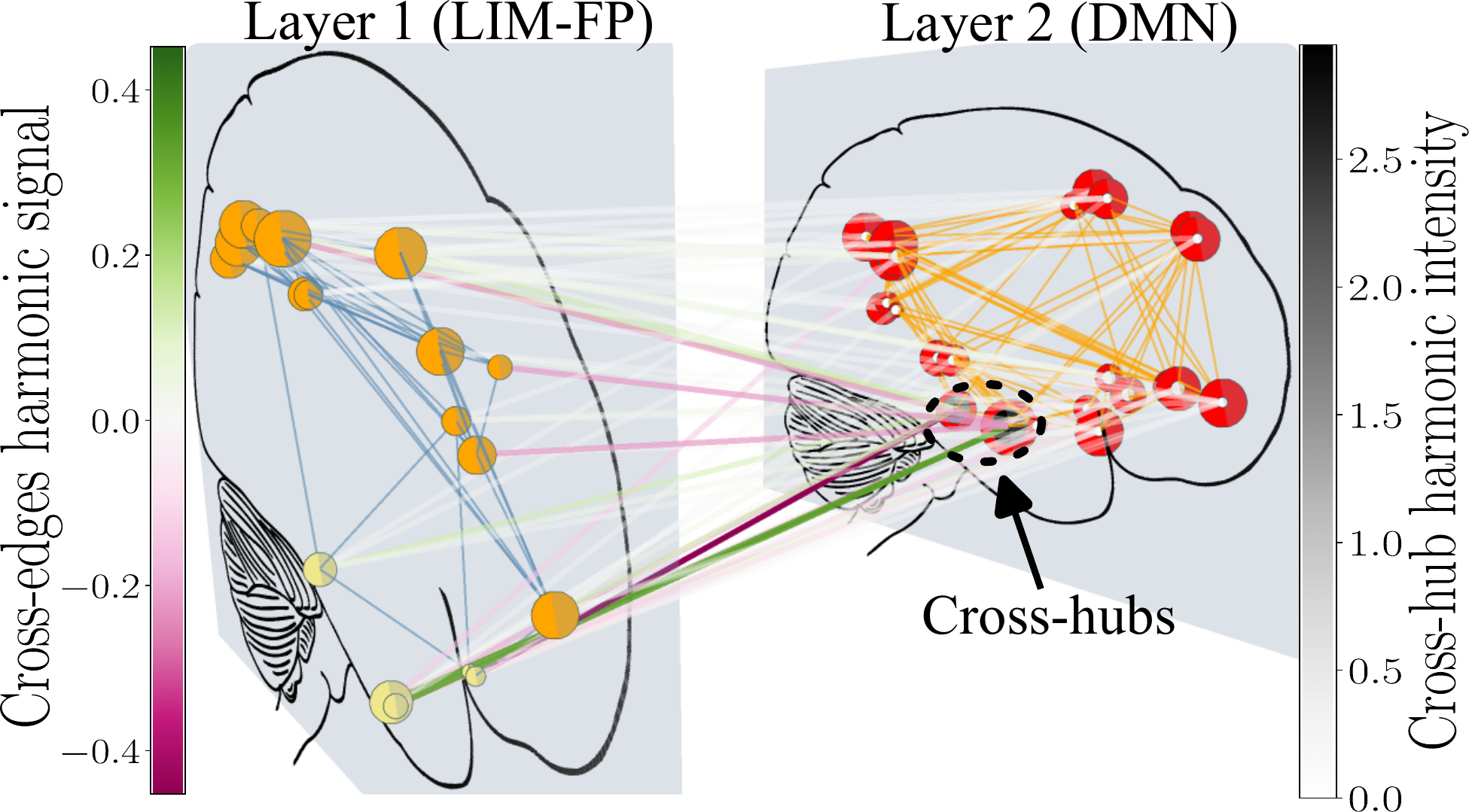}
        \caption{Cross-edge harmonic flows and cross-hub harmonic intensity.}
    \end{subfigure}
    \vspace{0.3cm}
    \begin{subfigure}[b]{\columnwidth}
    \centering
        \includegraphics[width=0.9\textwidth, height=4.8cm]{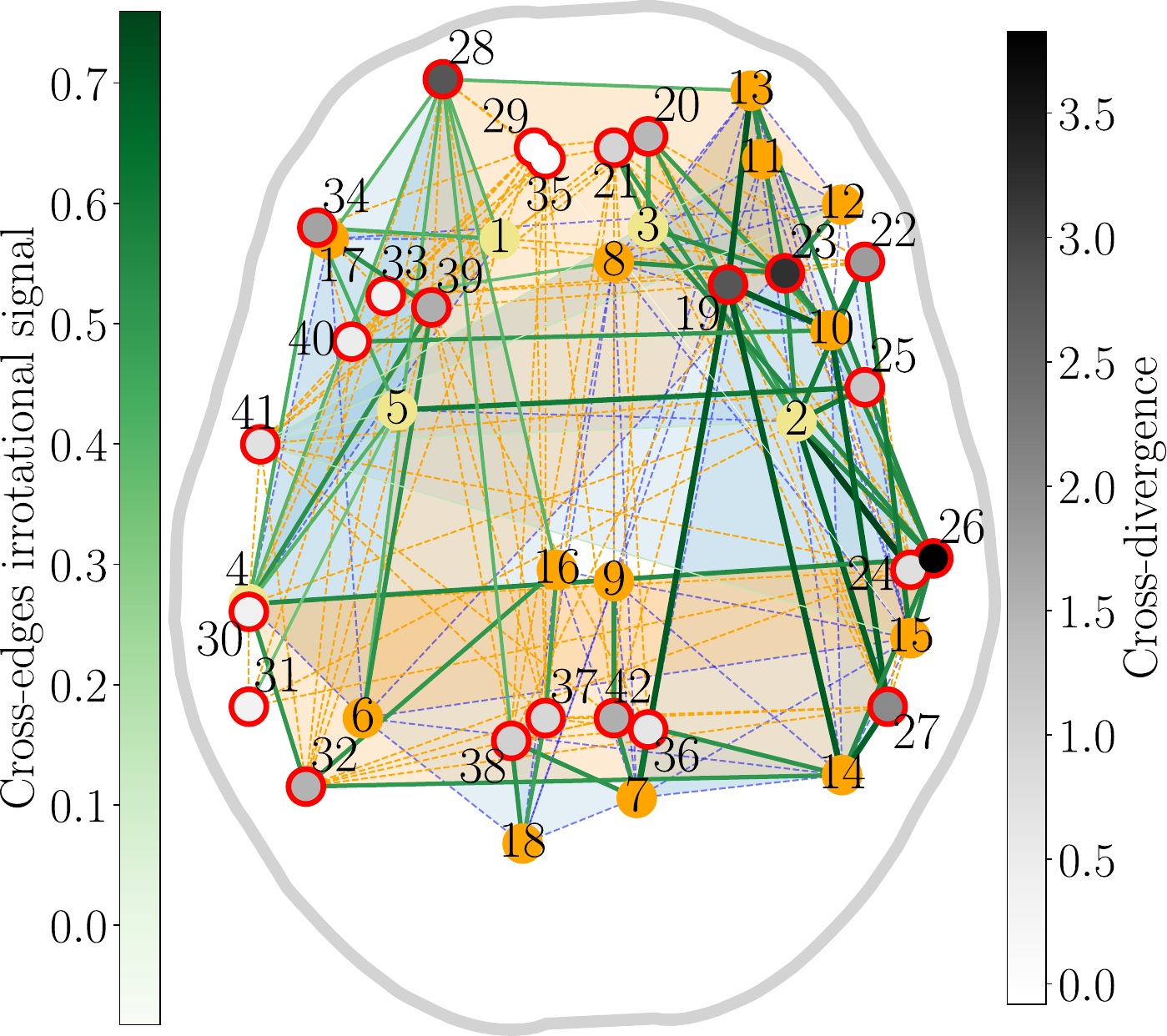}
        \caption{Cross-edges irrotational signal and node cross-divergence.}
    \end{subfigure}
    \vspace{0.3cm}
    \begin{subfigure}[b]{\columnwidth}
    \centering
        \includegraphics[width=0.88\textwidth, height=4.8cm]{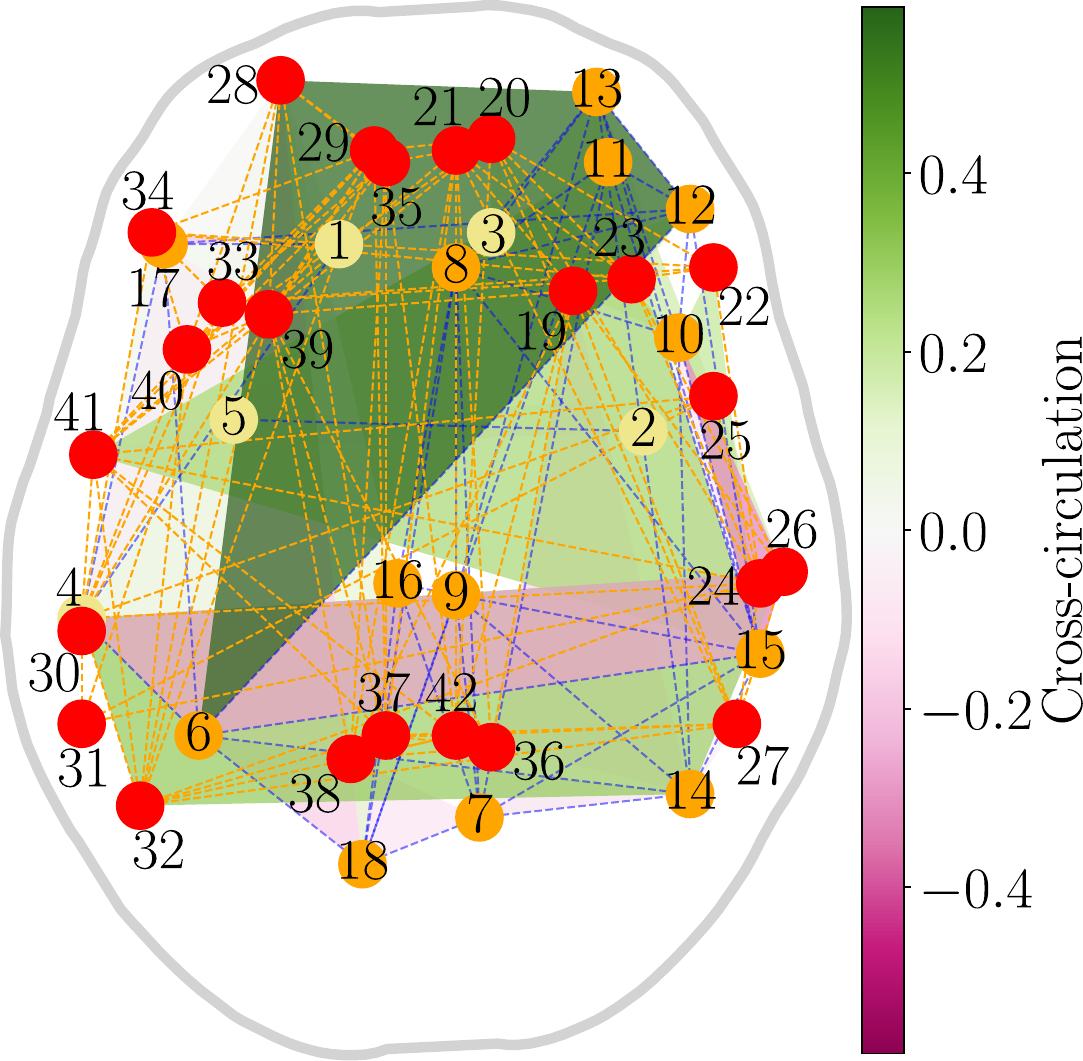}
          \caption{Cross-circulation observed over the $2$-order $(1,0)$ cross-cells.}
    \end{subfigure}
    
    \caption{Learned cross-cell complex between module $1$ and $2$  and recovered cross-edge signals.}
    \label{fig:brain_modules}
\end{figure}

Given that fMRI signals exhibit Gaussian behavior~\cite{hlinka}, we apply the graphical LASSO method to identify statistically significant edges~\cite{peel2022statistical}. This yields a functional connectivity matrix that captures both intra- and cross-module connections. For computational practical purposes, our illustrative analysis in this application is restricted to a brain CMC composed of two modules (as shown in Fig. \ref{fig:brain_modules}(a)), based on the schematic representation of a two-layer CMC as in Fig.~\ref{fig:CMC_2}. Our multilayer brain visualization setup were developed based on \cite{breedt2023multimodal}.
The first module consists of $N_1 = 18$ nodes from the Limbic (LIM) subnetwork, shown in khaki, and the Frontoparietal (FP) subnetwork, shown in orange. These nodes are interconnected by $N_{1,-1}^{1,2} = 45$ intra-module edges, depicted as blue lines. The second module includes $N_2 = 24$ nodes from the Default Mode Network (DMN), illustrated in red, connected by $N_{-1,1}^{1,2} = 92$ intra-module edges, shown as orange lines. The inter-module connectivity of the cross-layer graph is defined by $N_{0,0}^{1,2} = 66$ cross-edges, linking the nodes across layers. The diameter of each ROI reflects its node degree, and all ROIs are spatially projected onto a two-dimensional Euclidean space according to Schaefer's atlas coordinates~\cite{schaefer_atlas}. These subnets (DMN, FPN, LIM) were prioritized in this example due to their identification as core resting-state networks in large-scale parcellation studies~\cite{stephen, parcellation}, reflecting their interplay in baseline processes: self-referential cognition (DMN), cognitive flexibility (FP), and emotion integration (LIM). A topographical view of Fig.~\ref{fig:brain_modules}(a) is depicted in  Fig. 8(a) of Appendix C, with the solid outer gray contour delineating the brain’s boundaries. 
We learn the cross-complex structure by solving the optimization problem in  (\ref{eq:bas_pur_L00}) to infer
$N_{1,0}^{1,2}=34$  filled $(1,0)$ cross-cells. These are visualized  in Fig.~\ref{fig:brain_modules}(b) as 29 blue triangles and 5 orange quadrilaterals, forming the incidence matrix $\mB_{0,1}^{(1),2}$ between the two modules (see also the visualization in Fig. 8(a) of Appendix C). The cells $(1,0)$ have faces of order $1$ on layer $1$ and one node on layer $2$. We estimated the cross-Laplacian matrix \beq 
\hat{\mL}_{0,0}^{(1),2}= (\mB_{0,0}^{(1),2})^{T} \mB_{0,0}^{(1),2} +\hat{\mB}_{1,0}^{(1),2}(\hat{\mB}_{1,0}^{(1),2})^{T} 
\eeq
with cross-Betti number $\beta_{0,0}^{1,2}=9$, since we have $8$ open cones and one closed cone.  Hence, we found the spectral representation   of the harmonic, cross-irrotational and cross-solenoidal components of the observed time-series $\mX_{1}^{1,2}=[\mx_{1}^{1,2}(1),\mx_{1}^{1,2}(2),\ldots,\mx_{1}^{1,2}(M)]$ as
\beq
\begin{split}
&\hat{\mX}_{1,H}^{1,2}=(\hat{\mU}_{H,0}^{(1),2})^T{\mX}_{1}^{1,2}, \; \; {\mX}_{1,H}^{1,2}=\hat{\mU}_{H,0}^{(1),2}\hat{\mX}_{1,H}^{1,2}, \\ &\hat{\mX}_{1,d}^{1,2}=(\mU_{d,0}^{(1),2})^T{\mX}_{1}^{1,2}, \quad 
{\mX}_{1,d}^{1,2}=\mU_{d,0}^{(1),2}\hat{\mX}_{1,d}^{1,2}, \\ &\hat{\mX}_{1,u}^{1,2}=(\hat{\mU}_{u,0}^{(1),2})^T{\mX}_{1}^{1,2}, \quad {\mX}_{1,u}^{1,2}=\hat{\mU}_{u,0}^{(1),2}\hat{\mX}_{1,u}^{1,2}
\end{split}
\eeq
where $\hat{\mU}_{H,0}^{(1),2}$ are the eigenvectors spanning the kernel of $\hat{\mL}_{0,0}^{(1),2}$ (in this application of dimension $9$), $\mU_{d,0}^{(1),2}$ and $\hat{\mU}_{u,0}^{(1),2}$ are the eigenvectors associated with the non-zero eigenvalues of the Laplacians $\mathbf{{L}}_{d \, 0,0}^{(1),2}$ and $\mathbf{\hat{L}}_{u \, 0,0}^{(1),2}$, respectively. We averaged our results over $M=1000$  time series. In Fig. \ref{fig:brain_modules}(a), we report the cross-edge average harmonic component encoded by the solid edges in green-pink scale. The strongest harmonic flows are concentrated around the cones containing harmonic cross-hubs, i.e. the nodes $41, 26, 27, 21, 22, 34$ (as illustrated in Fig. 8(a) in Appendix  C). Further details on the associated brain ROIs and their labels are provided in  Table II  of Appendix  C. 
To quantify the strength of the harmonic cross-hubs, we compute for each hub $i$ the value $\text{CH}(i) = \frac{\sum_{j=1}^{n_i} |x_{1,H}(e_{ij})|}{\max(|\mathbf{x}_{1,H}|)}$, where $n_i$ denotes the number of cross-edges $e_{ij}$ that share the $i$-th hub and its surrounding cones.
Then, $\text{CH}(i)$ is the normalized sum of the intensities of the harmonic signals over the cross-edges of the cones having the hub as a common vertex. These values are visualized as grayscale dots over the $N_2$ nodes in Fig.~\ref{fig:brain_modules}(a) (with diameter proportional to their magnitude), and Fig. 8(a)  in Appendix  C. These findings indicate that specific nodes within the DMN play a key role in modulating or coordinating activity across networks in the first layer. Notably, nodes 41 and 26, located in the left and right temporal lobes, respectively (see Fig. 8(a) in Appendix  C), serve as controllers or connectors between regions of the FP and LIM subnetworks, which exhibit otherwise weak connectivity at the first layer, as shown in Fig.~\ref{fig:brain_modules}(a). This supports the interpretation that DMN nodes act as cross-network hubs, facilitating the integration of information between the FPN and LIM subnetworks. Such a role is consistent with the DMN’s known function as a connector hub in the brain’s network architecture~\cite{van2013network}, potentially enabling indirect FPN–LIM communication during resting-state. 

In Fig. \ref{fig:brain_modules}(b) we represent the cross-irrotational component and the cross-divergence on the nodes within the second  module, since we focus on the $(1,0)$ cross-cells. Interestingly, it can be noticed as almost all nodes in layer $2$ have a positive divergence. Specifically, nodes $26$ and $23$ exhibit the highest divergence values of $3.82$ and $3.20$, respectively. Node $35$ is the only sink node, having a small divergence value of $-0.082$. 
Therefore, the nodes exhibiting the highest magnitudes of cross-divergence can be interpreted as key sources or sinks of information flowing from layer 2 to layer 1. This interpretation aligns with the DMN recognized role in mediating information exchange between networks involved in cognitive control (FP) and emotion/memory processing (LIM)~\cite{menon2023dmn}.

Finally, in Fig. \ref{fig:brain_modules}(c) we illustrate the cross-circulation of the average cross-edge signal $\bar{{\mx}}_{1}^{1,2}=\sum_{i=1}^{M}{\mx}_{1}^{1,2}(i)/M$ derived as $\text{curl}_{cr}(\bar{{\mx}}_{1}^{1,2})= \hat{\mB}_{1,0}^{(1),2\, T} \bar{{\mx}}_{1}^{1,2}$ with $M=1000$. It can be observed that the  $(1,0)$ cross-cell identified by the nodes $(6,12,13,28)$  has the strongest cross-circulation, i.e it represents a quadruple of nodes from DMN and FP that are strongly coupled. 
This functional coupling is consistent with the integration of DMN and FPN during resting-state, which supports internally oriented cognitive processes such as autobiographical memory, planning, and self-referential thinking, as discussed in \cite{dixon2018heterogeneity}. Additionally, pronounced cross-circulation patterns involving nodes from the LIM, FPN, and DMN subnetworks, such as the most negative circulation observed over the cross-cell $(4,6,14,26)$, shown in Fig.~\ref{fig:brain_modules}(c) and  Fig. 8(b) in Appendix  C, suggest functional integration among these subsystems. This finding resonates with previous studies that emphasize the role of DMN in mediating interactions between emotions-related (LIM) and cognitive control-related (FP) processes~\cite{menon2023dmn, yeshurun2021default}. It is important to stress that this section is intended as an illustrative proof‑of‑concept of usability of the proposed framework and  we do not claim a definitive neuroscientific result.
Our aim is to demonstrate how the proposed cross‑Laplacian tools can expose integrative, higher‑order interactions between layers in a multilayer network; a full neurobiological interpretation deserves further investigation and it is beyond the scope of this theoretical contribution.



\vspace{-0.4cm}

\section{Conclusion}
 In this paper we present topological signal processing tools over cell multicomplexes spaces. We introduce CMCs that are novel topological spaces able to represent the intra- and inter-layer higher order connectivity among cell complexes. Then, we proposed an algebraic representation of CMCs  based on cross-Laplacian matrices which enables Hodge-based decompositions of the signals defined over these spaces. Hence, we extend topological signal processing tools to the analysis and processing of signals defined over CMCs. The developed framework enables  local processing of signals at different scales by considering  local homologies.
 We showed how this novel framework provides a powerful tool for detecting critical cross-hubs and cross-circulation patterns between brain subnetworks, offering a promising pathway for the development of interpretable markers of network integration across layers.
 In this paper we considered the $(0,0)$-cross-Laplacians, but future developments should  focus on  studying the homologies induced by $(k,n)$-cross-Laplacians  to fully characterize the local signal processing  over higher-order CMCs. 
 
\label{sec: conclusion}
\bibliographystyle{IEEEbib}
\bibliography{reference}

\begin{thebibliography}{10}

\bibitem{boccaletti2006complex}
S.~Boccaletti, V.~Latora, Y.~Moreno, M.~Chavez, and D.-U. Hwang,
\newblock ``Complex networks: Structure and dynamics,''
\newblock {\em Phys. rep.}, vol. 424, no. 4-5, pp. 175--308, 2006.

\bibitem{strogatz2001exploring}
S.~H. Strogatz,
\newblock ``Exploring complex networks,''
\newblock {\em Nature}, vol. 410, no. 6825, pp. 268--276, 2001.

\bibitem{kivela2014multilayer}
M.~Kivel{\"a}, A.~Arenas, M.~Barthelemy, J.~P. Gleeson, Y.~Moreno, and M.~A Porter,
\newblock ``Multilayer networks,''
\newblock {\em Journal of complex networks}, vol. 2, no. 3, pp. 203--271, 2014.

\bibitem{de2013mathematical}
M.~De~Domenico et~al.,
\newblock ``Mathematical formulation of multilayer networks,''
\newblock {\em Phys. Review X}, vol. 3, no. 4, pp. 041022, 2013.

\bibitem{bianconi2018multilayer}
G.~Bianconi,
\newblock {\em Multilayer networks: structure and function},
\newblock Oxford university press, 2018.

\bibitem{CRAINIC20221}
T.~G. Crainic, B.~Gendron, and M.~R. {Akhavan Kazemzadeh},
\newblock ``A taxonomy of multilayer network design and a survey of transportation and telecommunication applications,''
\newblock {\em Eur. J. Oper. Res.}, vol. 303, no. 1, pp. 1--13, 2022.

\bibitem{racz2024multilayer}
A.~R{\'a}cz-Szab{\'o}, T.~Ruppert, and J.~Abonyi,
\newblock ``Multilayer network-based evaluation of the efficiency and resilience of network flows,''
\newblock {\em Complexity}, vol. 2024, no. 1, pp. 6940097, 2024.

\bibitem{dickison2016multilayer}
M.~E. Dickison, M.~Magnani, and L.~Rossi,
\newblock {\em Multilayer social networks},
\newblock Cambridge University Press, 2016.

\bibitem{10.1093/gigascience/gix004}
M.~De~Domenico,
\newblock ``Multilayer modeling and analysis of human brain networks,''
\newblock {\em GigaScience}, vol. 6, no. 5, pp. gix004, 02 2017.

\bibitem{breedt2023multimodal}
L.~C. Breedt et~al.,
\newblock ``Multimodal multilayer network centrality relates to executive functioning,''
\newblock {\em Netw. Neurosci.}, vol. 7, no. 1, pp. 299--321, 2023.

\bibitem{liu2020robustness}
X.~Liu et~al.,
\newblock ``Robustness and lethality in multilayer biological molecular networks,''
\newblock {\em Nat. commun.}, vol. 11, no. 1, pp. 6043, 2020.

\bibitem{zhao2025constructing}
H.~Zhao, H.~Xu, T.~Wang, and G.~Liu,
\newblock ``Constructing multilayer {PPI} networks based on homologous proteins and integrating multiple {P}age{R}ank to identify essential proteins,''
\newblock {\em BMC bioinformatics}, vol. 26, no. 1, pp. 80, 2025.

\bibitem{krishnagopal2023topology}
S.~Krishnagopal and G.~Bianconi,
\newblock ``Topology and dynamics of higher-order multiplex networks,''
\newblock {\em Chaos, Solitons \& Fractals}, vol. 177, pp. 114296, 2023.

\bibitem{munkres2018elements}
J.~R. Munkres,
\newblock {\em Elements of algebraic topology},
\newblock CRC press, 2018.

\bibitem{Lim}
L.-H. Lim,
\newblock ``{H}odge {L}aplacians on graphs,''
\newblock {\em S. Mukherjee (Ed.), Geometry and Topology in Statistical Inference, Proc. Sympos. Appl. Math., 76, AMS}, 2015.

\bibitem{barb_2020}
S.~Barbarossa and S.~Sardellitti,
\newblock ``Topological signal processing over simplicial complexes,''
\newblock {\em IEEE Trans. Signal Process.}, vol. 68, pp. 2992--3007, March 2020.

\bibitem{sardellitti2024topological}
S.~Sardellitti and S.~Barbarossa,
\newblock ``Topological signal processing over generalized cell complexes,''
\newblock {\em IEEE Trans. Signal Process.}, vol. 72, pp. 687--700, 2024.

\bibitem{Roddenberry_23}
T.~M. Roddenberry, V.~P. Grande, F.~Frantzen, M.~T. Schaub, and S.~Segarra,
\newblock ``Signal processing on product spaces,''
\newblock in {\em IEEE Int. Conf. Acoust., Speech and Signal Process. (ICASSP)}, 2023, pp. 1--5.

\bibitem{moutuou2023}
E.~M. Moutuou, O.~B.~K. Ali, and H.~Benali,
\newblock ``Topology and spectral interconnectivities of higher-order multilayer networks,''
\newblock {\em Frontiers in Complex Systems}, vol. 1, pp. 1281714, 2023.

\bibitem{SardellittiCMC2024}
S.~Sardellitti, B.~C. Bispo, F.~A.~N. Santos, and J.~B. Lima,
\newblock ``Cross-{L}aplacians based topological signal processing over cell multicomplexes,''
\newblock in {\em 33rd European Signal Processing Conference (EUSIPCO) 2025, Palermo, Italy, Sept.}, 2025.

\bibitem{klette2000cell}
R.~Klette,
\newblock ``Cell complexes through time,''
\newblock in {\em Vision Geometry IX}. Int. Soc. for Opt. and Photon., 2000, vol. 4117, pp. 134--145.

\bibitem{barmak2011algebraic}
J.~A. Barmak,
\newblock {\em Algebraic topology of finite topological spaces and applications}, vol. 2032,
\newblock Springer, 2011.

\bibitem{grady2010}
L.~J. Grady and J.~R. Polimeni,
\newblock {\em Discrete calculus: Applied analysis on graphs for computational science},
\newblock Sprin. Scie. \& Busin. Media, 2010.

\bibitem{munkres2000topology}
J.~R. Munkres,
\newblock {\em Topology},
\newblock Prentice Hall, 2000.

\bibitem{goldberg2002combinatorial}
T.~E. Goldberg,
\newblock ``Combinatorial {L}aplacians of simplicial complexes,''
\newblock {\em Senior Thesis, Bard College}, 2002.

\bibitem{hatcher2005algebraic}
A.~Hatcher,
\newblock {\em Algebraic topology},
\newblock Cambr. Univ. Press, 2005.

\bibitem{Donoho98}
S.~S. Chen, D.~L. Donoho, and M.~A. Saunders,
\newblock ``Atomic decomposition by basis pursuit,''
\newblock in {\em SIAM J. Sci. Comput.}, 1998, vol.~20, pp. 33--61.

\bibitem{hoppe2024representing}
J.~Hoppe and M.~T. Schaub,
\newblock ``Representing edge flows on graphs via sparse cell complexes,''
\newblock in {\em Learning on Graphs Conference}. PMLR, 2024, pp. 1--1.

\bibitem{gurugubelli2024simplicial}
S.~Gurugubelli and S.~P. Chepuri,
\newblock ``Simplicial complex learning from edge flows via sparse clique sampling,''
\newblock in {\em 32nd European Signal Processing Conference (EUSIPCO)}, 2024, pp. 2332--2336.

\bibitem{bullmore2009complex}
E.~Bullmore and O.~Sporns,
\newblock ``Complex brain networks: graph theoretical analysis of structural and functional systems,''
\newblock {\em Nature reviews neuroscience}, vol. 10, no. 3, pp. 186--198, 2009.

\bibitem{arbabyazd2020dynamic}
L.~M. Arbabyazd, D.~Lombardo, O.~Blin, M.~Didic, D.~Battaglia, and V.~Jirsa,
\newblock ``Dynamic functional connectivity as a complex random walk: definitions and the d{FC}walk toolbox,''
\newblock {\em MethodsX}, vol. 7, pp. 101168, 2020.

\bibitem{HCP}
D.~C. {Van Essen}, S.~M. Smith, D.~M. Barch, T.~E.~J. Behrens, E.~Yacoub, and K.~Ugurbil,
\newblock ``The {WU}-{M}inn human connectome project: An overview,''
\newblock {\em NeuroImage}, vol. 80, pp. 62--79, 2013.

\bibitem{hlinka}
J.~Hlinka, M.~Paluš, M.~Vejmelka, D.~Mantini, and M.~Corbetta,
\newblock ``Functional connectivity in resting-state f{MRI}: Is linear correlation sufficient?,''
\newblock {\em NeuroImage}, vol. 54, no. 3, pp. 2218--2225, 2011.

\bibitem{peel2022statistical}
L.~Peel, T.~P. Peixoto, and M.~De~Domenico,
\newblock ``Statistical inference links data and theory in network science,''
\newblock {\em Nature Communications}, vol. 13, no. 1, pp. 6794, 2022.

\bibitem{schaefer_atlas}
A.~Schaefer et~al.,
\newblock ``Local-global parcellation of the human cerebral cortex from intrinsic functional connectivity {MRI},''
\newblock {\em Cerebral cortex (New York, N.Y. 1991)}, vol. 28, no. 9, pp. 3095--3114, 2018.

\bibitem{stephen}
S.~M. Smith et~al.,
\newblock ``Correspondence of the brain's functional architecture during activation and rest,''
\newblock {\em Proceedings of the National Academy of Sciences}, vol. 106, no. 31, pp. 13040--13045, 2009.

\bibitem{parcellation}
B.~T.~T. Yeo et~al.,
\newblock ``The organization of the human cerebral cortex estimated by intrinsic functional connectivity,''
\newblock {\em Journal of Neurophysiology}, vol. 106, no. 3, pp. 1125--1165, 2011.

\bibitem{van2013network}
M.~P. van~den Heuvel and O.~Sporns,
\newblock ``Network hubs in the human brain,''
\newblock {\em Trends in Cognitive Sciences}, vol. 17, no. 12, pp. 683--696, 2013.

\bibitem{menon2023dmn}
V.~Menon,
\newblock ``20 years of the default mode network: A review and synthesis,''
\newblock {\em Neuron}, vol. 111, no. 16, pp. 2469--2487, 2023.

\bibitem{dixon2018heterogeneity}
M.~L. Dixon et~al.,
\newblock ``Heterogeneity within the frontoparietal control network and its relationship to the default and dorsal attention networks,''
\newblock vol. 115, no. 7, pp. E1598--E1607, 2018.

\bibitem{yeshurun2021default}
Y.~Yeshurun, M.~Nguyen, and U.~Hasson,
\newblock ``The default mode network: where the idiosyncratic self meets the shared social world,''
\newblock {\em Nature Reviews Neuroscience}, vol. 22, no. 3, pp. 181--192, Mar. 2021.

\end{thebibliography}
\appendices{}
\vspace{0.3cm}
\section{Proof of Proposition $1$}
\label{sec:AppendixA}
 Given the (monolayer) cell complex $\mathcal{X}$, let us  denote by $\mB_q$ the boundary   matrix describing the incidence between cells of order $q-1$ and $q$ as in (1). It is known \cite{Lim} that the boundary of a boundary is zero, i.e. it holds \beq \label{eq:BqBq+1}\mB_q \mB_{q+1}=\mathbf{0}. \eeq
   It straightly follows  the local orthogonality condition:
    \beq \label{eq:BqkBq+1i}
    \mB_q(p,:) \mB_{q+1}(:,i)=0, \; \forall\, p,i
    \eeq
    where: the $i$-th column vector $\mB_{q+1}(:,i)$ identifies with its  entries (1 or -1) the $q$-cells $c_{q}(j)$ lower bounding $c_{q+1}(i)$, while   
    the entries of the $p$-th row vector $\mB_q(p,:)$ identifies the cells $c_{q}(j)$ upper bounding $c_{q-1}(p)$.      
    Given the indexes $p$ and $i$, equation (\ref{eq:BqkBq+1i}) can be  rewritten in the form
    \beq \label{eq:local_cond}
    \sum_{j \, : \, c_{q-1}(p)\prec_b c_{q}(j)\prec_b c_{q+1}(i)}\hspace{-1cm} 
    \text{sign}(c_{q-1}(p),c_{q}(j))\text{sign}(c_{q}(j),c_{q+1}(i))=0 
    \eeq
    where $\text{sign}(c_{q}(j),c_{q+1}(i))=1$ (or -1) if $c_{q}(j) \sim c_{q+1}(i) $ or $c_{q}(j) \nsim c_{q+1}(i)$.
    Therefore, given the cell $c_{q-1}(p)$  
    the condition of orthogonality  in (\ref{eq:BqkBq+1i}) is satisfied locally  for all cells upper bounding $c_{q-1}(p)$ and lower bounding $c_{q+1}(i)$.
    This implies that the local orthogonality condition is also verified by the cells in the cross- and intra-layers cell complexes. Specifically, considering the local boundary matrix in  (6) (or, equivalently, in (7)), from (\ref{eq:local_cond}) we get
   \beq \label{eq:local_cond1}
   \underset{j \, : \, \underset{\normalsize{c_{q}^{\ell,m}(j)\prec_b c_{q+1}^{\ell,m}(i)}}{\normalsize{c_{q-1}^{\ell,m}(p)\prec_b c_{q}^{\ell,m}(j)}}}{\sum}\hspace{-0.7cm}  \text{sign}( c_{q-1}^{\ell,m}(p),c_{q}^{\ell,m}(j)) \text{sign}(c_{q}^{\ell,m}(j),c_{q+1}^{\ell,m}(i))
   =0 
    \eeq
    with $c_{q+1}^{\ell,m}(i) \in \mathcal{X}_{k,n}^{\ell,m}$, $c_{q}^{\ell,m}(j) \in \mathcal{X}_{k-1,n}^{\ell,m}$ and $c_{q-1}^{\ell,m}(p) \in \mathcal{X}_{k-2,n}^{\ell,m}$. Therefore, (\ref{eq:local_cond1}) can be expressed in the form
    \beq \label{eq:BqkBq+1i2}
    \mB_{k,n}^{(\ell),m}(p,:) \mB_{k+1,n}^{(\ell),m}(:,i)=0, \; \forall p,i.
    \eeq
    This proves the condition i) in (11) and using similar derivations, we can also prove the second condition ii).

\vspace{0.8cm}
\section{Proof of Theorem $1$}
\label{sec:AppendixB}
Let us consider the cross-Betti number $\beta_{0,0}^{(\ell)}=\text{dim}(\text{ker}(\mL_{0,0}^{(\ell),m}))$. Focusing on a second order CMC, from (15), we get 
\beq \label{eq: Hodge_space_dec_00}
\mathbb{R}^{N_{0,0}^{(\ell),m} }\equiv \text{img}((\mB_{0,0}^{(\ell),m})^{T}) \oplus \text{ker}(\mL_{0,0}^{(\ell),m}) \oplus  \text{img}(\mB_{1,0}^{(\ell),m}), 
\eeq
so that  we easily derive
\beq \label{eq: betti_00}
\beta_{0,0}^{(\ell)}=N_{0,0}^{\ell,m} - \text{rank}((\mB_{0,0}^{(\ell),m})^{T}) - \text{rank}(\mB_{1,0}^{(\ell),m}).
\eeq
Note that the rank of the $N_{-1,0}^{\ell,m} \times N_{0,0}^{\ell,m}$ matrix $\mB_{0,0}^{(\ell),m}$ satisfies the equality  $\text{rank}((\mB_{0,0}^{(\ell),m})^{T})=\text{rank}(\mB_{0,0}^{(\ell),m}(\mB_{0,0}^{(\ell),m})^{T})$.  The matrix $\mB_{0,0}^{(\ell),m}(\mB_{0,0}^{(\ell),m})^{T}$ is a diagonal matrix with entries the node upper-cross degrees, then  \beq 
\label{eq:rank1}\text{rank}((\mB_{0,0}^{(\ell),m})^{T})=\text{rank}(\mB_{0,0}^{(\ell),m}(\mB_{0,0}^{(\ell),m})^{T})=n_0^{m} \eeq denoting with $n_0^{m}$ the number of nodes on layer $m$ that are connected with cross-edges.
Let us now derive the rank of
the $N_{0,0}^{\ell,m} \times N_{1,0}^{\ell,m}$ matrix $\mB_{1,0}^{(\ell),m}$. This matrix  has $N_{1,0}^{\ell,m}$ independent columns, then its rank is equal to  
\beq \label{eq:rank2}
\text{rank}(\mB_{1,0}^{(\ell),m})=\min(N_{1,0}^{\ell,m},N_{0,0}^{\ell,m})
\eeq
where  $N_{1,0}^{\ell,m}$ is the number  of $2$-order (filled) independent cross-cells with vertex on layer $m$ and sides  on layer $\ell$. Let us denote with $n_{0,0}^{m}(k)$ the number of cross-edges incident to the node $k$ on layer $m$ and  with $n_{c}^{m}(k)$ the number of possible cones (closed or open) incident to the node $k$. For each node, considering a $2$-order CMC, we  have $n^{m}_{0,0}(k)-1$ possible   independent convex wedges with $2$ cross-edges as boundaries. Since these wedges can be cones or  $2$-order cross-cells we get 
\beq n_{0,0}^{m}(k)-1=n_{c}^{m}(k)+N_{1,0}^{\ell, m}(k) \eeq
 where $N_{1,0}^{\ell, m}(k)$ is the number of cross-cells incident to node $k$ on layer $m$.
 Then, it holds
\beq N_{1,0}^{\ell, m}(k)=n_{0,0}^{m}(k)-1-n_{c}^{m}(k) \eeq
and summing on the nodes
we get
\beq \sum_{k=1}^{n_0^m}N_{1,0}^{\ell, m}(k)=\sum_{k=1}^{n_0^m} n_{0,0}^{m}(k)-n_0^m-\sum_{k=1}^{n_0^m} n_{c}^{m}(k) \eeq
i.e. 
\beq \label{eq: N_10}
N_{1,0}^{\ell, m}=N_{0,0}^{\ell,m}-n_0^m- n_{c}^{m} \geq 0. \eeq
From this last equality it follows that $N_{1,0}^{\ell, m}\leq N_{0,0}^{\ell,m}$, so that from (\ref{eq:rank2}) we get \beq \label{eq:rank_B10} \text{rank}(\mB_{1,0}^{(\ell),m})=N_{1,0}^{\ell,m}. \eeq
Replacing equations (\ref{eq:rank1}) and (\ref{eq:rank_B10}) in (\ref{eq: betti_00}), we get
\beq \label{eq: betti_00_1}
\beta_{0,0}^{(\ell)}=N_{0,0}^{\ell,m} - n_0^{m} -  N_{1,0}^{\ell, m}
\eeq
and  using the equality in (\ref{eq: N_10}), it holds 
\beq \label{eq: betti_00_1}
\beta_{0,0}^{(\ell)}=n_{c}^{m}.
\eeq
This proves that the Betti number $\beta_{0,0}^{(\ell)}$ returns the number of $(1,0)$ cones between layers $\ell,m$. Similar derivations hold for Betti number $\beta_{0,0}^{(m)}$.
\vspace{1cm}
\section{Supplementary Figures and Table}

\begin{figure*}[t]
    \centering
   
    \begin{subfigure}[b]{\columnwidth}
    \centering
        \includegraphics[width=0.9\textwidth]{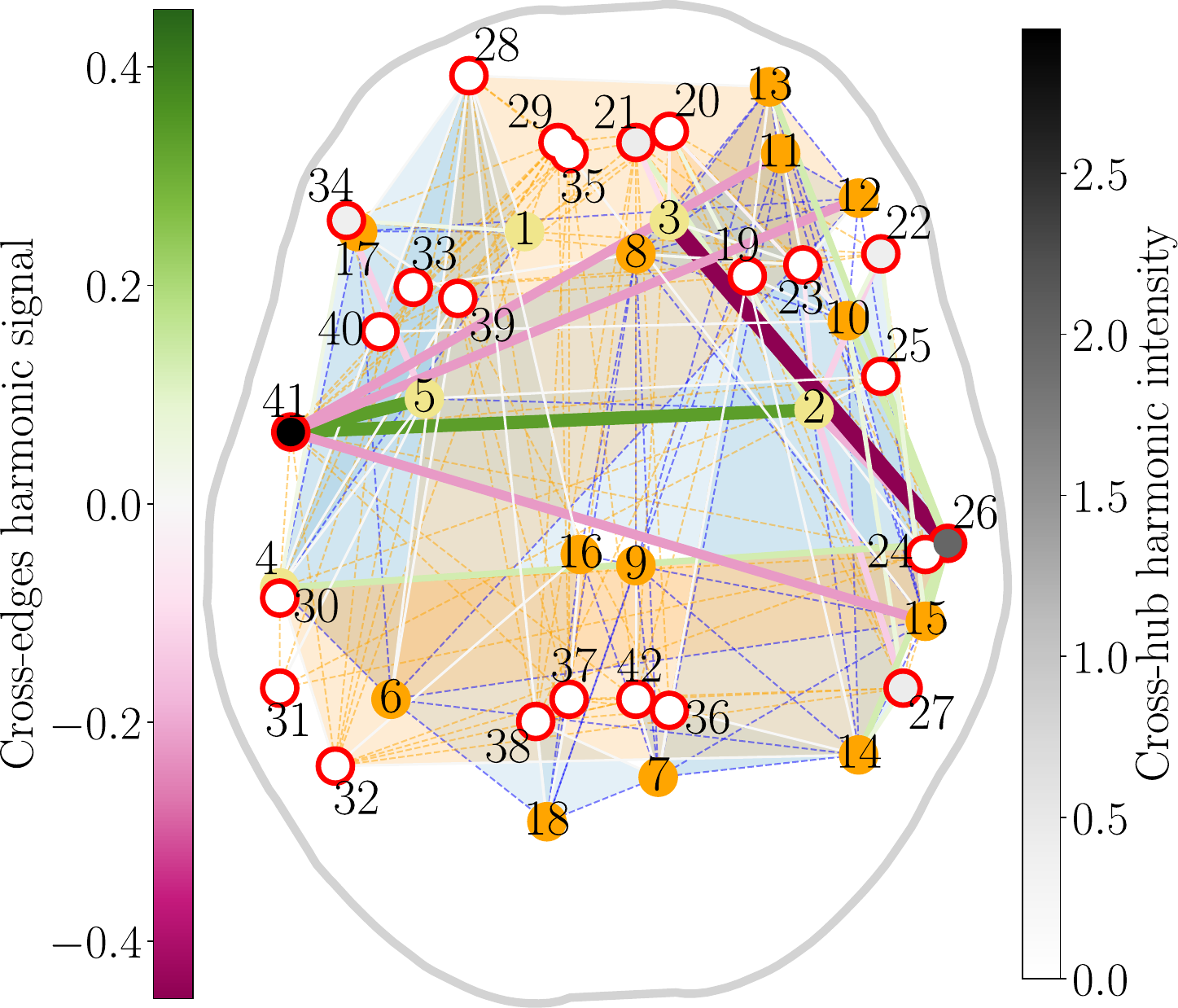}
        \caption{Cross-edge harmonic signal and cross-hub harmonic intensity.}
    \end{subfigure}
    \begin{subfigure}[b]{\columnwidth}
    \centering
        \includegraphics[width=0.9\textwidth]{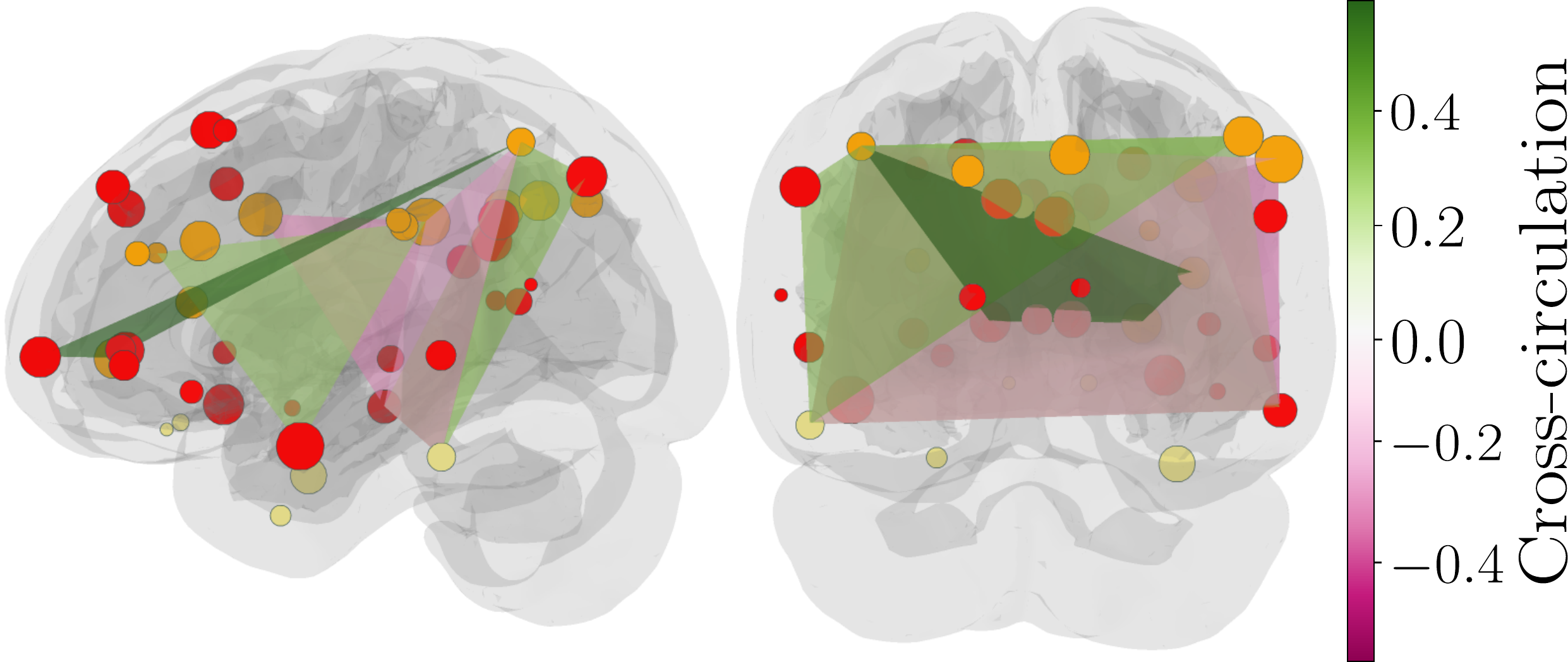}
        \caption{Strongest cross-circulations observed over the $2$-order $(1,0)$ cross-cells. A side view of the brain is shown on the left, and a rear view is presented on the right.}
    \end{subfigure}
    
    \caption{Supplementary figures providing alternative perspectives of Fig. 7.}
    \label{fig:brain_modules_supplementary}
\end{figure*}

\begin{table*}[h]
\centering
\caption{Information of the brain regions based on the Schaefer's brain parcellation~\cite{schaefer_atlas}.}
\begin{tabular}{|c|c|c|c|}
\hline
\textbf{Layer} & \textbf{Subnetwork} & \textbf{Region name} & \textbf{Label} \\
\hline
\multirow{18}{*}{Layer 1} 
  & \multirow{5}{*}{LIM} 
  & L\_Limbic\_Orbitofrontal\_Cortex\_1 & 1 \\
  &  & R\_Limbic\_TempPole\_1              & 2 \\
  &  & R\_Limbic\_Orbitofrontal\_Cortex\_1 & 3 \\
  &  & L\_Limbic\_TempPole\_2              & 4 \\
  &  & L\_Limbic\_TempPole\_1              & 5 \\
\cline{2-4}
& \multirow{13}{*}{FP} 
  & L\_Cont\_Par\_1                     & 6 \\
  &  & R\_Cont\_Precuneus\_1              & 7 \\
  &  & R\_Cont\_Prefrontal\_Cortex\_m\_p\_1 & 8 \\
  &  & R\_Cont\_Cingulate\_1              & 9 \\
  &  & R\_Cont\_Prefrontal\_Cortex\_l\_4   & 10 \\
  &  & R\_Cont\_Prefrontal\_Cortex\_l\_3   & 11 \\
  &  & R\_Cont\_Prefrontal\_Cortex\_l\_2   & 12 \\
  &  & R\_Cont\_Prefrontal\_Cortex\_l\_1   & 13 \\
  &  & R\_Cont\_Par\_2                     & 14 \\
  &  & R\_Cont\_Par\_1                     & 15 \\
  &  & L\_Cont\_Cingulate\_1              & 16 \\
  &  & L\_Cont\_Prefrontal\_Cortex\_l\_1   & 17 \\
  &  & L\_Cont\_Precuneus\_1              & 18 \\
\hline
\multirow{24}{*}{Layer 2} & \multirow{24}{*}{DMN} 
  & R\_Default\_Prefrontal\_Cortex\_d\_Prefrontal\_Cortex\_m\_3 & 19 \\
  &  & R\_Default\_Prefrontal\_Cortex\_d\_Prefrontal\_Cortex\_m\_2 & 20 \\
  &  & R\_Default\_Prefrontal\_Cortex\_d\_Prefrontal\_Cortex\_m\_1 & 21 \\
  &  & R\_Default\_Prefrontal\_Cortex\_v\_2                       & 22 \\
  &  & R\_Default\_Prefrontal\_Cortex\_v\_1                       & 23 \\
  &  & R\_Default\_Temp\_3                                       & 24 \\
  &  & R\_Default\_Temp\_2                                       & 25 \\
  &  & R\_Default\_Temp\_1                                       & 26 \\
  &  & R\_Default\_Par\_1                                        & 27 \\
  &  & L\_Default\_Prefrontal\_Cortex\_4                         & 28 \\
  &  & L\_Default\_Prefrontal\_Cortex\_5                         & 29 \\
  &  & L\_Default\_Temp\_2                                       & 30 \\
  &  & L\_Default\_Par\_1                                        & 31 \\
  &  & L\_Default\_Par\_2                                        & 32 \\
  &  & L\_Default\_Prefrontal\_Cortex\_1                         & 33 \\
  &  & L\_Default\_Prefrontal\_Cortex\_2                         & 34 \\
  &  & L\_Default\_Prefrontal\_Cortex\_3                         & 35 \\
  &  & R\_Default\_Precuneus\_Posterior\_Cingulate\_Cortex\_1    & 36 \\
  &  & L\_Default\_Precuneus\_Posterior\_Cingulate\_Cortex\_2    & 37 \\
  &  & L\_Default\_Precuneus\_Posterior\_Cingulate\_Cortex\_1    & 38 \\
  &  & L\_Default\_Prefrontal\_Cortex\_7                         & 39 \\
  &  & L\_Default\_Prefrontal\_Cortex\_6                         & 40 \\
  &  & L\_Default\_Temp\_1                                       & 41 \\
  &  & R\_Default\_Precuneus\_Posterior\_Cingulate\_Cortex\_2    & 42 \\
\hline
\end{tabular}
\label{tab:regions_info}
\end{table*}

\end{document}